\documentclass[11pt,letterpaper]{article}

\usepackage{amsmath,amsthm,amssymb} %
\usepackage{authblk} %
\usepackage[english]{babel} %
\usepackage{bbm} %
\usepackage{booktabs} %
\usepackage{color} %
\usepackage{derivative} %
\usepackage{dsfont} %
\usepackage{enumitem} %
\usepackage{float} %
\usepackage{flushend} %
\usepackage[T1]{fontenc} %
\usepackage{framed} %
\usepackage[top=1in,bottom=1in,left=1.3in,right=1.3in]{geometry} %
\usepackage{graphicx} %
\usepackage[backref,colorlinks,citecolor=blue,bookmarks=true,breaklinks=true]{hyperref} %
\usepackage[capitalize]{cleveref} 
\usepackage{hyphenat} %
\usepackage{ifdraft} %
\usepackage[utf8]{inputenc} %
\usepackage{interval} %
\usepackage{mathrsfs} %
\usepackage{mathtools} %
\usepackage{mdframed} %
\usepackage{microtype} %
\usepackage{physics} %
\usepackage{soul} %
\usepackage{times} %
\usepackage{tikz} %
\usepackage{xparse} %
\usepackage{xspace} %

\definecolor{Gray}{rgb}{0.85,0.85,.85}
\colorlet{ChannelColor}{Gray}

\usepackage{tikz} %
\usetikzlibrary{
  decorations.pathreplacing,
  arrows,
  angles,
  quotes,
  calc,
  calligraphy,
  decorations.pathmorphing,
  positioning,
  shapes}

\theoremstyle{plain} 
\newtheorem{theorem}{Theorem}[section] %
\newtheorem{lemma}{Lemma}[section] %
\newtheorem{proposition}{Proposition}[section] %
\newtheorem{corollary}{Corollary}[section] %
\theoremstyle{definition} %
\newtheorem{definition}{Definition}[section] %
\theoremstyle{remark} %
\newtheorem*{remark}{Remark} %

\newcommand{\PT}{\textnormal{\textsf{P}}\xspace}
\newcommand{\NP}{\textnormal{\textsf{NP}}\xspace}
\newcommand{\FPT}{\textnormal{\textsf{FPT}}\xspace}

\newcommand{\XP}{\textnormal{\textsf{XP}}\xspace}

\newcommand{\W}{\textnormal{\textsf{W}}\xspace}
\newcommand{\BQP}{\textnormal{\textsf{BQP}}\xspace}

\newcommand{\QMA}{\textnormal{\textsf{QMA}}\xspace}
\newcommand{\FPQT}{\textnormal{\textsf{FPQT}}\xspace}

\newcommand{\QW}{\textnormal{\textsf{QW}}\xspace}

\newcommand{\QWP}{\textnormal{\textsf{QW$[\text{\textsf{P}}]$}}\xspace}

\newcommand{\SQW}{\textnormal{\textsf{SQW}}\xspace}
\newcommand{\ETH}{\textnormal{\textsf{ETH}}\xspace}
\newcommand{\QETH}{\textnormal{\textsf{QETH}}\xspace}
\newcommand{\QCETH}{\textnormal{\textsf{QCETH}}\xspace}
\newcommand{\M}{\textnormal{\textsf{M}}\xspace}
\newcommand{\QM}{\textnormal{\textsf{QM}}\xspace}
\newcommand{\poly}{\operatorname{poly}}

\renewcommand{\real}{\mathbb{R}}
\newcommand{\microspace}{\mspace{.5mu}} %
\renewcommand{\ket}[1]{\ensuremath{\lvert\microspace#1%
    \microspace\rangle}} %
\newcommand{\bigket}[1]{\bigl\lvert\microspace#1%
  \microspace\bigr\rangle} %
\newcommand{\Bigket}[1]{\Bigl\lvert\microspace#1%
  \microspace\Bigr\rangle} %
\renewcommand{\bra}[1]{\ensuremath{\langle\microspace#1%
    \microspace\rvert}} %
\newcommand{\bigbra}[1]{\bigl\langle\microspace#1%
  \microspace\bigr\rvert} %
\newcommand{\Bigbra}[1]{\Bigl\langle\microspace#1%
  \microspace\Bigr\rvert} %

\setlist[description]{font={\it}, leftmargin=2.2cm, style=nextline, itemsep=0pt,
  topsep=5pt}

\graphicspath{{graphics/}}

\newcommand{\problem}[1]{\textsc{#1}}

\begin{document}

\title{\Large\bf Parameterized Complexity of Weighted Local Hamiltonian Problems
  and the Quantum Exponential Time Hypothesis}
\author[1,2,3]{Michael J. Bremner\thanks{michael.bremner@uts.edu.au}}
\author[4]{Zhengfeng Ji\thanks{jizhengfeng@tsinghua.edu.cn}}
\author[4]{Xingjian Li\thanks{lxj22@mails.tsinghua.edu.cn}}
\author[2,3]{Luke Mathieson\thanks{luke.mathieson@uts.edu.au}}
\author[1,2,3]{\hspace{5em} Mauro E.S. Morales\thanks{mauricio.moralessoler@student.uts.edu.au}}
\affil[1]{Centre for Quantum Computation and Communication Technology}
\affil[2]{Centre for Quantum Software and Information}
\affil[3]{School of Computer Science, Faculty of Engineering \& Information Technology, University of Technology Sydney, NSW 2007, Australia}
\affil[4]{Department of Computer Science and Technology, Tsinghua University, Beijing, China}

\maketitle

\begin{abstract}
  We study a parameterized version of the local Hamiltonian problem, called the
  weighted local Hamiltonian problem, where the relevant quantum states are
  superpositions of computational basis states of Hamming weight $k$.
  The Hamming weight constraint can have a physical interpretation as a
  constraint on the number of excitations allowed or particle number in a system.
  We prove that this problem is in $\QW[1]$, the first level of the quantum weft
  hierarchy, and that it is hard for $\QM[1]$, the quantum analogue of $\M[1]$.
  Our results show that this problem cannot be fixed parameter quantum tractable
  (FPQT) unless certain natural quantum analogue of the exponential time
  hypothesis (ETH) is false.
\end{abstract}

\section{Introduction}

Parameterized complexity theory~\cite{downey_fundamentals_2013} aims to analyze
problems in a more refined manner than classical complexity theory by creating
tools for comparing complexity over multiple parameters, as opposed to simply
the input size.
In principle, this opens up more possibilities for the analysis of real-world
problems that may have complicated parameter dependencies.
The key idea of the theory is to confine the possible super-polynomial
dependence of the runtime to the parameters only.
Tractability in the parameterized setting is described by the fixed parameter
tractable class \FPT, which can be thought of as a parameterized version of
$\PT$ where the dependence on the parameter of the runtime can be any computable
function while the dependence on the input size is still a polynomial.

One of the most useful results in parameterized complexity theory is that
problems which are complete for the class $\W[1]$, one of many natural
parameterized generalizations of \NP{}, are not fixed-parameter tractable unless
the exponential time hypothesis (\ETH) fails (see~\cite{Chen2006W1ETH} for the
initial result,~\cite{downey_fundamentals_2013} for further exposition
and~\cite{Impagliazzo1999eth,Impagliazzo2001eth} where \ETH{} is defined).
More concisely, if $\W[1]=\FPT$ then $\ETH$ is false.
This links parameterized intractability to classical intractability, tying the
intractability of SAT with the parameterized intractability of $\W[1]$-complete
problems such as \problem{$k$-Independent Set} and \problem{$k$-Clique}.

The \problem{Local Hamiltonian Problem}~\cite{kitaev2002classical} has been one
of the most studied problems in \emph{quantum} complexity theory over the last
two decades.
There have been many interesting recent works studying variants of this problem
with relevance for quantum chemistry.
In~\cite{Gorman2021electronicstructure}, the authors establish the
\QMA-completeness of a variant of the local Hamiltonian problem considering a
fixed basis describing the orbitals of the electronic structure problem,
inspired by the problem posed in~\cite{Whitfield2013electronic}.
Another work in this direction is that of~\cite{Gharibian2021dequantizing},
where the authors study the so called \problem{Guided Local Hamiltonian Problem}
in which the instance description includes a local Hamiltonian $H$ and a state
vector $u$ promised to be close to the ground state of $H$.
In this work it is shown that when the Hamiltonian is $6$-local then the
decision problem is \BQP-hard, further
work~\cite{Gharibian2022improved,Cade2022complexityguided} has shown that the
problem remains \BQP-hard when considering 2-local Hamiltonians.

In this paper we link the complexity of the \problem{Weight-$k$ $\ell$-Local
  Hamiltonian} problem to the classical \ETH and quantum variants of it, \QETH{}
and \QCETH{}.
It is shown that if the \problem{Weight-$k$ $\ell$-Local Hamiltonian} problem can
be solved in \FPT or \FPQT (the quantum generalization of \FPT introduced
in~\cite{bremner2022parameterized}) then versions of these hypotheses will fail.
The weight in this problem refers to the Hamming weight of the states in the
promise of the local Hamiltonian problem, either there is a weight-$k$ state
with a small eigenvalue, or all weight-$k$ states are above a certain energy.
The restriction of the weight on the states considered in the problem finds a
physical interpretation when considering the $1$s in the computational basis as
particle excitations and thus the weight corresponds to fixing the particle
number to $k$.
We remark that when considering Fermionic Hamiltonians, using the Jordan-Wigner
transform in general makes the Hamiltonian non-local, in this sense our results
should be considered as a first step towards the goal of studying the complexity
of problems with fixed particle number.

Mirroring the situation with the parameterized complexity class $\W[1]$, we
establish our main result, that the \problem{Weight-$k$ $\ell$-Local Hamiltonian}
is contained in $\QW[1]$, a quantum generalization of \W[1] that was introduced
in~\cite{bremner2022parameterized}.
These classes are parameterized analogues of \NP{} and \QMA{} where the promise
has bounded weight and the verifier is limited to the use of ``weft-1'' circuits
-- circuits that have at most a single ``large'' gate on any path from the input
to the output.
In the quantum case this is defined to be a single multiply controlled Toffoli
gate acting on $O(n)$ qubits.
Analogous to the classical case, the link to the exponential-time hypotheses is
made via introducing a miniaturized version of the circuit satisfiability
problem that is used to define a class $\QM[1]$ and proving that the
\problem{Weight-$k$ $\ell$-Local Hamiltonian} is $\QM[1]$-hard.

By establishing that the \problem{Weight-$k$ $\ell$-Local Hamiltonian} is contained
in $\QW[1]$ and that it can be linked to the quantum exponential time
hypotheses, we have resolved a clear question that emerged
from~\cite{bremner2022parameterized}, whether or not there was evidence of any
problems in $\QW[1]$ that were likely outside both $\W[1]$ and \FPQT.

We believe that there are a number of interesting open problems emerging from
our work.
The first is that it remains an open question as to whether this problem is
complete for $\QW[1]$, or even if there are any natural problems at all that are
complete for this class.
One of the roadblocks to proving this is the absence of a \emph{normalization}
theory for quantum weft circuits.
In an attempt to take the first steps towards resolving this problem, we
identify a class inside $\QW[1]$ (called $\SQW_{1}[1]$) that has a normal form
and we prove that the \problem{Frustration-Free Weighted Hamiltonian} problem is
complete for that class.

From a physical perspective, in this paper we have established the likely
intractability of the weighted local Hamiltonian problem in an attempt to better
connect parameterized complexity theory to problems of interest in quantum
chemistry.
Moving forward it would be of significant interest if there are additional
structures for this problem that would place it in either \FPQT or \FPT.
A natural variation to consider would be to further restrict the locality
conditions on the Hamiltonian, for example to consider lattice problems or other
regular graph structures.
Another direction would be to further restrict the set of potential promises on
this problem.

\subsection{Summary of Main results}\label{sec:main-results}

Our main contribution in this paper is to initiate the study of local
Hamiltonian problems in the context of parameterized complexity theory.
In classical complexity theory, many important parameterized problems such as
\problem{$k$-Vertex Cover}, \problem{$k$-Independent Set}, and
\problem{Weight-$k$ SAT} consider problems parameterized by the Hamming weight
of the solution.
Local Hamiltonian problems are natural generalizations of constraint
satisfiability problems (CSP) and it is natural to study the complexity of the
local Hamiltonian problem restricting the solution Hilbert space to the span of
basis strings of Hamming weight $k$ (for the definition of weight see
\cref{def:weight-state} and for the weighted local Hamiltonian see
\cref{def:weight-LH}).
This setup has some physical relevance as the weight $k$ may be interpreted as
the particle number or the number of excitations in a physical system.

We establish the likely intractability of the weighted local Hamiltonian problem.
Our first result puts this constrained-weight local Hamiltonian problem in
$\QW[1]$, a quantum version of the classical parameterized intractability class
$\W[1]$, which uses ``weft-$1$'' verifiers, these are circuits of constant depth
and at most one large gate (see \cref{def:weft-quantum,def:QW} for more
details).

\begin{theorem}[Informal version of \cref{thm:weighted-ham}]%
  \label{thm:informal-weighted-ham}
    \problem{Weight-$k$ $\ell$-Local Hamiltonian} is in $\QW[1]$.
\end{theorem}

We have not been able to prove the $\QW[1]$-completeness of this problem.
It is known that this problem is contained in $\XP$ in contrast to the weighted
quantum circuit satisfiability problem which cannot be in $\XP$ unless
$\PT=\BQP$~\cite{bremner2022parameterized}.
In~\cite{bremner2022parameterized} it is shown that $\QW[1]$ is not a subset of
$\XP$ unless $\PT = \BQP$, however as it is also shown that $\FPQT$ is not a
subset of $\W[1]$ under the same assumption, it may be that although
\problem{Weight-$k$ $\ell$-Local Hamiltonian} is in $\XP$, its closure under $\FPQT$
reductions could be $\QW[1]$ (i.e., it could be $\QW[1]$-complete).

\emph{Nonetheless}, the proof technique leveraged in
\cref{thm:informal-weighted-ham} enabled us to show that the weighted local
Hamiltonian problems are $\QM[1]$-hard, where $\QM[1]$ is the natural quantum
analogue of $\M[1]$, a class introduced in classical parameterized complexity
theory to develop intractability results in reference to parameterized
sub-exponential time algorithms, and thence to the exponential time hypothesis
(ETH).
This establishes the intractability of the weighted local Hamiltonian problem
under quantum analogues of the ETH we introduce in this paper.

In this work we consider a \emph{weak} version of the \ETH as given in
\cite{downey_fundamentals_2013} (see \cref{def:eth}).
This version states that there are no classical algorithms solving the circuit
satisfiability problem in subexponential time.
Moreover, we establish that quantum parameterized complexity is connected to two
quantum generalizations of the \ETH which we define as $\QCETH$ and $\QETH$ (See
\cref{def:QCETH,def:QETH}).
The first variant, $\QCETH$, says that there are no quantum algorithms solving
classical circuit satisfiability in subexponential time, while $\QETH$ roughly
states that no such subexponential quantum algorithms exist for a version of the
quantum circuit satisfiability problem.
We show that, as expected, \QCETH implies \QETH (see
\cref{prop:QCETH-imp-QETH}).

In \cref{sec:qw-eth} we make these connections explicit, first by recalling a
theorem known in classical parameterized complexity which states that if
$\W[1]=\FPT$ then \ETH is false.
A simple reduction from \problem{Independent Set} to the \problem{Weighted Local
  Hamiltonian} problem shows that the existence of \FPT algorithms for the local
Hamiltonian problem would also lead to refuting \ETH, though we remark there
could be \FPT algorithms under different, but still interesting,
parameterizations.
Moreover, we prove the following result concerning \QETH.

\begin{theorem}[Informal version of
  \cref{thm:LH-QETH}]\label{thm:informal-LH-qeth}
  If \problem{Weight-$k$ $\ell$-Local Hamiltonian} is in \FPQT then $\QETH$ is
  false.
\end{theorem}

\Cref{thm:informal-LH-qeth} gives evidence to the claim that the weighted local
Hamiltonian is intractable as depicted in \cref{fig:param-class}.
Moreover, as stated in \cref{thm:informal-LH-qeth}, the tractability of the
weighted local Hamiltonian problem for quantum algorithms has implications
regarding the existence of subexponential quantum algorithms for a quantum
\QMA-complete problem.
As \QETH is implied by \QCETH, \cref{thm:informal-LH-qeth} implies that if
\problem{Weight-$k$ $\ell$-Local Hamiltonian} is in \FPQT then $\QCETH$ is
false.

\begin{figure}[ht]
  \centering

  \def\firstellipse{(0,0) ellipse (1 and 1)}
  \def\secondellipse{(0,.3) ellipse (1.2 and 1.5)}
  \def\thirdellipse{(0,.6) ellipse (1.5 and 2)}
  \def\fourthellipse{(0,1.4) ellipse (2 and 3)}

  \begin{tikzpicture}
    \filldraw[fill=gray!20]\fourthellipse;
    \filldraw[fill=gray!10]\thirdellipse;
    \filldraw[fill=gray!5]\secondellipse;
    \filldraw[fill=white]\firstellipse;
    \node at (0,0) {\FPT{}};
    \node at (0,1.25) {\FPQT{}};
    \node at (0,2.15) {$\QM[1]$};
    \node at (0, 3.85) {$\QW[1]$};

    \node[circle, fill=black, minimum width=5pt, inner sep=0] (p) at (0, 3.05) {};

    \node[font=\scriptsize] (problem) at (-2,4.8)
    {\problem{Weight-$k$ Local Hamiltonian}};
    \draw[->, shorten >=2pt, shorten <=2pt, thick] (problem) -- (p);

    \node[font=\scriptsize, inner sep=0] (ETH1) at (5,3)
    {\problem{Weight-$k$ Local Hamiltonian} $\in \FPQT$};
    \node[font=\scriptsize, node distance=0pt, right=of ETH1] (F1)
    {$\Rightarrow \QETH$ is false};
    \node[font=\scriptsize, node distance=1.2em, below=of F1.west,
    anchor=west] (F2) {$\Rightarrow \QCETH$ is false};

    \node[font=\scriptsize, inner sep=0, node distance=2, below=of ETH1.west,
    anchor=west] (ETH2) {\problem{Weight-$k$ Local Hamiltonian} $\in \FPT
      \Rightarrow \ETH \text{ is false}$};

  \end{tikzpicture}
  \caption{Summary of main results in our work.
    We prove that \problem{Weight-$k$ $\ell$-Local Hamiltonian} is in $\QW[1]$
    (previously it was proven to be in \XP~\cite{bremner2022parameterized}).
    Moreover, we show that if this problem lies in \FPT, then \ETH is false and
    if it is in \FPQT this implies that both $\QETH$ and $\QCETH$ are
    false.}\label{fig:param-class}
\end{figure}
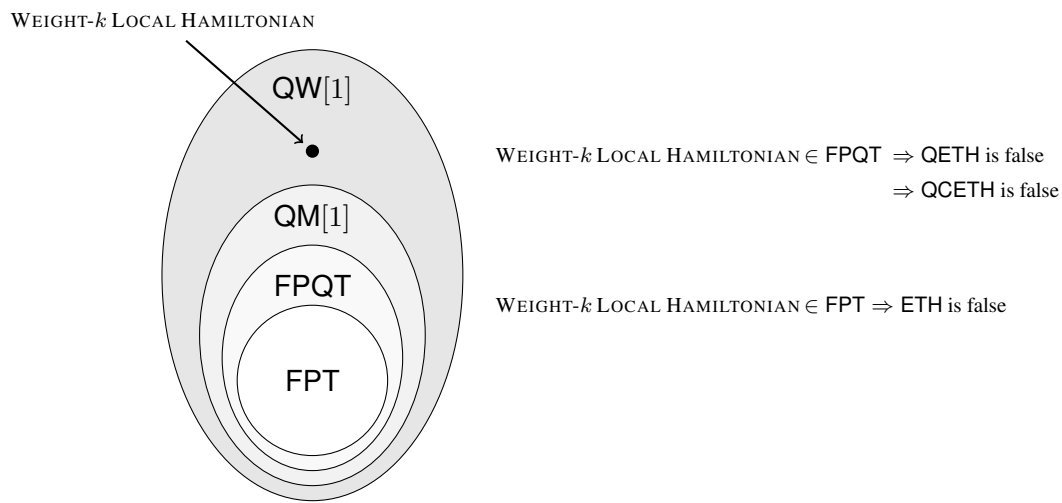

Finally, we also consider the \problem{Frustration-Free Weighted Hamiltonian}
problems and the equivalent \problem{Weighted Quantum SAT} problems.
Reusing the reductions in the proof of \cref{thm:informal-weighted-ham}, we are
able to show that frustration-free weighted local Hamiltonian problems are
contained in $\QW_{1}[1]$, the counterpart of $\QW[1]$ with perfect completeness
condition.
In fact, we prove a stronger statement that the problem is in $\SQW_{1}[1]$, a
class contained in $\QW_{1}[1]$ that models the weft-$1$ quantum circuit with
the large gate being the final AND gate.
Furthermore, we are able to show that the \problem{Frustration-Free Weighted
  Hamiltonian} and \problem{Weighted Quantum SAT} problems are actually complete
for the class $\SQW_{1}[1]$.

\begin{theorem}[Informal version of
  \cref{thm:SQW1-Frustration-QSAT}]\label{thm:informal-SQW}
  The \problem{Frustration-Free Weight-$k$ Hamiltonian} and \problem{Weight-$k$
    Quantum SAT} problems are $\SQW_{1}[1]$-complete.
\end{theorem}

The completeness of the weighted quantum SAT problems for $\SQW_{1}[1]$ raises
the questions like whether $\SQW_{1}[1]$ equals $\QW_{1}[1]$ and whether
$\SQW[1]$ equals $\QW[1]$.
The resolution of these questions relies on further investigations on
whether there is a quantum equivalent to the classical normalization theorems
which holds for quantum weft-$1$
circuits.
We leave it as an interesting open problem.

\subsection{Proof Techniques}

It is very natural to expect that, as the well-known quantum analogues of
constraint satisfaction problems, local Hamiltonian problems parameterized
using the weight are related to quantum analogues of the weft hierarchy.
Yet, even though classical \problem{Weight-$k$ SAT} problems are trivially
contained in $\W[1]$, to show that \problem{Weight-$k$ local Hamiltonian}
problems are in $\QW[1]$ is not easy.

The main technical challenge in proving this result is that, by definition, any
problem in $\QW[1]$ is characterized by being verifiable by a weft-$1$ circuit,
which implies we must check if the energy terms are satisfied using a circuit
with constrained depth.
A direct approach is by measuring the energy of each Hamiltonian term and then
either summing up or taking the average of the local energies.
However, this does not work as quantum measurements may perturb the witness
state, rendering later measurements problematic.
Even if we ignore this issue and assume somehow we can measure the local
energies without disturbing the state (probably like in the case of a quantum
SAT problem and we have perfect completeness), it is not easy to perform these
measurements in constant depth.
Classically, we can use fanout gates to wire the inputs to each checking term,
yet this is not possible in the quantum case!
Our solution to this is a long chain of reductions (shown in
\cref{fig:reductions}) that gradually normalizes the form of the Hamiltonian.

In the first step, we reduce the local Hamiltonian problem to a
\emph{weight-preserving} circuit satisfiability problem.
A weight-preserving quantum circuit is a circuit that consists of only unitary
gates that preserve the weight of any state.
It is the key technical tool we rely on in this chain of reductions.
On one hand, we can use weight-preserving circuits to perform energy
measurements for arbitrary Hamiltonian problems, including those whose local
terms are not weight-preserving operators (see
\cref{sec:weight-preserving-qcsat}).
On the other hand, we can convert the weight-preserving circuit problem back to
a weighted local Hamiltonian problem using Kitaev's
construction~\cite{kitaev2002classical,kempe2006complexity}.
As the history state of the Kitaev construction has the form
\begin{equation*}
  \ket{\psi_{\text{history}}} = \frac{1}{\sqrt{T+1}} \sum_{t=0}^{T}
  \bigl( U_{t} U_{t-1} \cdots U_{1} \ket{\psi}_{\text{witness}} \bigr)
  \otimes \ket{t}_{\text{clock}},
\end{equation*}
the weight-preserving property is essential to maintain the weight condition of
the resulting Hamiltonian problem.
If the circuit is weight-preserving and we choose to use the indicator clock (of
the form $\ket{0\cdots 010\cdots 0}$ where the $1$ is at $t$-th position), the
history state is a state of weight $k+1$.
The Kitaev-type quantum proof checking technique is crucial
for checking the propagation without blowing up the weight and without
resorting to any normalization theorem about the circuit.
If we were checking the propagation classically, the weight will be multiplied
by a factor of $T$, an overhead we won't be able afford.
Another advantage that the weight preserving circuit and Kitaev constructions
bring is that we do not need to have constant depth circuits yet, providing a
lot of flexibility to manipulate the circuit structure and the error reduction
we need to perform.

To the best of our knowledge, however, no previous work did a systematic study
on weight-preserving circuits.
Therefore, we first prove several basic facts about weight-preserving quantum
circuits.
Following the standard approach~\cite{Barenco1995gates} that established the
universal gate sets for standard quantum circuits, we show in
\cref{sec:universal} that there is similar theory of universal gate sets for
weight preserving circuits.
The result is of independent interest and may be useful elsewhere.

As the second step in the chain of reductions, we also show how to perform
strong $\QMA$ completeness and soundness error reduction for weight-preserving
verification circuits.
For standard quantum circuits, this is first proved by Marriott and
Watrous~\cite{marriott2005quantum}.
The Marriott and Watrous procedure utilizes the fact that the post-measurement
state of the verifier still contains ``most'' of the information about the
witness state and it is possible to perform a series of measurements with
outcomes $y_{1}, y_{2}, \ldots, y_{N}$ so that $z_{i} = y_{i-1} \oplus y_{i}$ are
identically independently distributed according to the acceptance probability.
Thus by doing simple statistics over $z_{i}$, we can reduce the error
exponentially.
This standard Marriott-Watrous construction does not seem to work in the weight
preserving setting.
The reason is that it is very hard to store or to ``forget'' the measurement
outcomes $y_{i}$.
For a weight-preserving circuit starting from a weight-$k$ state, the size of
the space it can possibly explore is at most $O(n^{k})$ dimensional and it is
incapable of retaining all measurement outcomes $y_{i}$.
There is also no easy way to ``forget'' about the $y_{i}$ information as quantum
computing is reversible and, in order to perform counting over $z_{i}$, we need
to remember the previous outcome $y_{i-1}$.
Fortunately for us, it is possible to adapt the fast $\QMA$ error reduction
based on quantum singular value transformation~\cite{gilyen2019quantum} in our
setting.
This is a rare example where the standard Marriott-Watrous error reduction won't
work while the fast $\QMA$ error reduction will.
What we utilize here is not that the fast $\QMA$ error reduction is
quadratically faster, but that it does not perform many measurements and that it
only uses a single extra qubit and easy-to-implement extra gates.

In the third step, we construct a weighted local Hamiltonian from the
weight-preserving circuit SAT problem.
Compared to the weighted local Hamiltonian problem we started with, this new
Hamiltonian problem has the nice property that the terms are (almost) spatially
sparse, meaning that each qubit is involved in a constant number of terms and it
is possible to partition the Hamiltonian into a finite number of groups of
non-overlapping terms.
This paved the way to constructing the final constant-depth weft-$1$ quantum
circuit.
This step is made possible by the spatially sparse construction introduced
in~\cite{OT08}.
There are two important changes we made to this construction in~\cite{OT08}.
First, we need to use \emph{indicator} clock to ensure the weight of the ground
state of the Hamiltonian is fixed to some number.
The checking of the indicator clock is much harder than the usual unary clock
and it will \emph{violate} the spatially sparse condition for the clock register.
We will see that this won't be a problem as for the checking of the clock, we
can measure all the clock qubits and perform a classical $\W[1]$ computation on
the measurement outcome.
Second, we introduced another simple encoding on the computational qubits,
mapping $\ket{0}$ to $\ket{00}$ and $\ket{1}$ to $\ket{11}$ for each qubit in
the computational register.
This ensures that the weight of the computational qubits is always an even
number.
So when we ask for a weight-$(2k+1)$ in the Hamiltonian problem, we do need to
check that the clock register is not in the all $0$ state $\ket{0^{T+1}}$.
This simple trick simplifies the checking of the indicator clock so that we only
need a big AND gate.
Otherwise, to rule out the all $0$ case, we needed a big OR gate and check that
the output is $1$.
This modification is crucial later on for obtaining the completeness result of a
weighted quantum satisfiability problem.

We note that due to a technical requirement in the next step, we needed a
stronger gap condition on the energy bounds $a,b$ of the local Hamiltonian.
The standard condition on the energy bounds says that $b-a$ is necessarily at
least $1/\poly(n)$.
We will require that $b/n^{2} - a$ is also at least $1/\poly(n)$.
This strong gap condition is made possible by the weight-preserving strong error
reduction procedure (given in \cref{sec:strong-error-reduction}).

In the final step of the proof that weighted local Hamiltonian is in $\QW[1]$,
we construct constant depth ``weft-$1$'' circuits for the almost spatially
sparse Hamiltonian problems we get from the third step.
This is now an easy step given all the preparations already performed.
The idea is that the circuit can measure constant number of groups of
non-overlapping terms in parallel.
To check the clock register, the circuit measures it and makes the decision
using classical fanout gates, NAND gates of fanin-$2$ and a classical big AND
gate.
Because we are measuring multiple Hamiltonian terms simultaneously, it is
necessary to consider the approximation between Hamiltonian of the form
$\sum_{i=1}^{n} \ket{0}\bra{0}_{i}$ and the project
$I - \ket{1^{n}}\bra{1^{n}}$.
This is why we need the stronger gap condition in the previous step.

\begin{figure}[ht]
  \centering
  \begin{tikzpicture}[node distance=.7cm,
    arr/.style={->, thick, shorten >= 2pt, shorten <= 2pt},
    box/.style={draw, rounded corners=5, fill=gray!10, align=center, font=\scriptsize,
      minimum height=3.5em, text width=4.5cm, inner sep=6}]

    \node[box] (MQCSAT) {\problem{Miniaturized Quantum
        Circuit Satisfiability}};
    \node[box, below=of MQCSAT] (WWPQCS) {\problem{Weight-$k$ Weight-Preserving\\
        Quantum Circuit Satisfiability}};
    \node[box, left=of WWPQCS] (WLH) {\problem{Weight-$k$ $\ell$-Local
        Hamiltonian}};
    \node[box, below=of WWPQCS] (WPSER) {\problem{Error-Reduced\\
        Weight-$k$ Weight-Preserving Quantum Circuit Satisfiability}};
    \node[box, below=of WPSER] (QW1) {\problem{Weight-$k$ Weft-$1$ Depth-$d$\\
        Quantum Circuit Satisfiability}};
    \node[box, left=of WPSER] (WLASSH) {\problem{Weight-$k$ $\ell$-Local
        Almost\\ Spatially Sparse Hamiltonian}};

    \draw[arr] (WLH) -- (WWPQCS);
    \draw[arr] (WWPQCS) -- (WPSER);
    \draw[arr] (WPSER) -- (WLASSH);
    \draw[arr] (WLASSH.south east) -- (QW1.north west);
    \draw[arr] (MQCSAT) -- (WWPQCS);
    \draw[arr] (WLASSH) -- (WLH);

  \end{tikzpicture}
  \caption{Reductions used to prove that the weighted local Hamiltonian problem
    is in $\QW[1]$ and $\QM[1]$-hard.
    Hamiltonian problems are the left while quantum circuit problems are on the
    right.}
  \label{fig:reductions}
\end{figure}

Having shown that the \problem{Weight-$k$ local Hamiltonian} is in $\QW[1]$, a
natural question is whether this problem is complete.
The main motivation to prove $\QW[1]$-completeness is to establish the likely
intractability of the problem.
Despite not showing this, we manage to give evidence of intractability in
\cref{thm:informal-LH-qeth}.
A proof sketch of this theorem is as follows.
First, a ``miniaturized'' version of the quantum circuit satisfiability problem
is defined where the instance is described by a weft-$t$ quantum circuit with at
most $k\log n$ inputs and gates, which we call \problem{Mini-QCSAT}.
From this problem we define the $\QM$-hierarchy as those problems
\FPQT-reducible to \problem{Mini-QCSAT} and show that $\QM[1]\subseteq\FPQT$ if
and only if $\QETH$ is false.
The next step is to prove that \problem{Weight-$k$ $\ell$-Local Hamiltonian} is
$\QM[1]$-hard.
For this, we reduce \problem{Mini-QCSAT} to \problem{Weight-$k$
  Weight-Preserving Quantum Circuit Satisfiability}, the proof of this reduction
is inspired by the so called ``$k\log n$'' trick used in classical parameterized
complexity (see Corollary 3.13 in~\cite{flum2006parameterized}).
More precisely, we reduce the miniaturized quantum circuit over $k\log n$ qubits
to a quantum circuit over $O(k n)$ qubits by encoding the $k\log n$ qubits in
$k$ groups of $\log n$ qubits and then further encode these with $k$ groups of
$n$ qubits with weight-$1$.
This last reduction allows us to leverage the reduction (Step 2 of the proof)
from \cref{thm:informal-weighted-ham} to reduce the weight-preserving circuit
satisfiability to weight-$k$ circuit satisfiability with constant-depth.
This implies that if the weighted Local Hamiltonian problem is in \FPQT then
\QETH is false, which further implies that \QCETH is false.

We also explored the complexity of the weighted quantum SAT problems and
frustration-free weighted quantum Hamiltonian problems.
This corresponds to the special case of the discussion on the containment of
weighted local Hamiltonian in $\QW[1]$ where we enforce perfect completeness.
It is easy to verify that along the chain of the reductions in
\cref{fig:reductions}, the perfect completeness condition is always kept and, as
we only used a big AND gate in the $\QW[1]$ verification circuit, this proves
that the weighted quantum SAT problems are in $\SQW_{1}[1]$, a special case of
$\QW_{1}[1]$ with the additional structural requirement that the big gate is the
last gate and it is the classical AND gate acting on measurement outcomes of a
constant-depth quantum circuit.

What is interesting is that for the quantum SAT problems, we are able to prove
that they are complete for $\SQW_{1}[1]$ even though the same idea fails for the
local Hamiltonian problems.
For this, we use a light cone argument to show that it is possible to directly
read off a set of local projectors forming a quantum SAT problem from the
constant depth $\SQW_{1}[1]$ circuit.
The light cone argument contracts the difference between classical fanout and
quantum fanout gates.
For an arbitrary classical fanout gate, the number of gates and wires needed to
be included in the light cone is always a constant while for a quantum fanout,
this number will include all qubits involved in the fanout gate and thus
unbounded.

\subsection{Organization}

The organization of this work is as follows.
\cref{sec:background} introduces the main parameterized complexity classes used
in our work together with the notations used.
\cref{sec:LH-qw1} provides the proof that the weighted local Hamiltonian problem
is in $\QW[1]$.
In \cref{sec:ff} some results on the Frustration Free Local Hamiltonian problems
are given and \cref{sec:qw-eth} presents the results regarding the Exponential
Time Hypothesis and the \QW-hierarchy.
Finally, we give some discussion of open problems and future directions in
\cref{sec:discussion}

\section{Background and Notation}\label{sec:background}

Here we present the main parameterized complexity classes defined
in~\cite{bremner2022parameterized} together with some important problems that
will be used.
For a more extended discussion of these definitions we direct the reader
to~\cite{bremner2022parameterized} with regard the quantum parameterized classes
and for a discussion of classical parameterized complexity we
suggest~\cite{downey1995,downey1999parameterized,downey_fundamentals_2013}.
We begin by recalling the definition of a parameterized problem, these are
problems where the description of the instance includes a parameter describing
certain property of the instance.

\begin{definition}[Parameterization]
  A parameterization of a finite alphabet $\Sigma$ is a mapping
  $\kappa:\Sigma^*\to\mathbb{Z}^+$ that is polynomial-time computable.
  The trivial parameterization $\kappa_\text{trivial}$ is the parameterization
  with $\kappa_\text{trivial}(x)=1$ for all $x\in\Sigma^*$.
\end{definition}
  
We now define a \emph{parameterized problem}.
\begin{definition}[Parameterized problem]
  A parameterized problem over a finite alphabet $\Sigma$ is a pair $(L,\kappa)$
  where $L\subseteq\Sigma^*$ is a set of strings over $\Sigma$ and $\kappa$ is a
  parameterization of $\Sigma$.
  We say that a parameterized problem $(L,\kappa)$ over the alphabet $\Sigma$ is
  \emph{trivial} if either $L=\varnothing$ or $L=\Sigma^*$.
\end{definition}

The complexity class of tractable bounded-error quantum parameterized problems
is \FPQT{} (see~\cite{bremner2022parameterized} for a formal definition).
While in quantum complexity theory the class \QMA is considered as the ``quantum
version'' of \NP, quantum parameterized complexity has many analogues of \NP (as
in the classical parameterized case).
The analogues of \NP we will focus on here are the classes \QWP and the those in
\QW-hierarchy~\cite{bremner2022parameterized}.
Here we give a modified definition of the \QW-hierarchy adequate for proving our
results.
The notion of Hamming weight is fundamental in parameterized complexity to
define intractability.
We will base our definitions in the following notion of weight for quantum
states.

\begin{definition}[Weight of a quantum state]\label{def:weight-state}
  A quantum state $\ket{\psi}=\sum_{x\in\{0,1\}^n}\alpha_x\ket{x}$ on $n$ qubits
  is said to have \emph{weight} $k$ if $\alpha_x=0$ for all $x$ not of Hamming
  weight $k$.
\end{definition}

A second notion important for defining the intractable $\W$-hierarchy in the
classical case is that of weft.

\begin{definition}[Circuit weft]\label{def:weft-classical}
  Given a Boolean circuit $\mathcal{C}$ comprising generalised Toffoli gates and
  one and two bit fan-in gates.
  The \emph{weft} of $\mathcal{C}$ is the maximum number of Toffoli gates that
  act on any path from input bit to output bit.
\end{definition}

This notion of weft is generalized to the quantum case:

\begin{definition}[Quantum circuit weft]\label{def:weft-quantum} 
  Given a quantum circuit $\mathcal{C}$ comprising generalised Toffoli gates,
  one and two-qubit gates, and unbounded \emph{classical} fanout.
  The \emph{weft} of $\mathcal{C}$ is the maximum number of Toffoli gates that
  act on any path from input qubit to output qubit.
\end{definition}

We remark that the fanout gate allowed in a weft-$1$ quantum circuit is
classical.
In a quantum circuit, a fanout gate is called classical if all of the target
qubits are initialized to the $\ket{0}$ state and no other gates acted on them
before the fanout gate.
After the fanout gate, a unitary gate can only act on the fanout qubits by using
them as controls.
The equivalence between this definition of classical fanout gates and the
standard definition follows from the principle of delayed measurements.
Because quantum fanout gates are very powerful and can simulate big Toffoli and
threshold gates~\cite{Hoyer2005_fanout}, they should be avoided when defining
weft-$t$ quantum circuits.

To define the \QW-hierarchy we proceed similarly as in
\cite{marriott2005quantum} for the class \QMA.
For functions $c,s:\mathbb{N}\to [0,1]$ we define the following problem

\begin{definition}[\problem{Weight-$k$ Weft-$t$ Depth-$d$ Quantum Circuit Satisfiability
  $(c,s)$}]\hfill
  \begin{description}
    \item[Instance:] A weft-$t$ depth-$d$ quantum circuit $\mathcal{C}$ on $n$
          witness qubits and $\poly(n)$ ancilla qubits.
    \item[Parameter:] A natural number $k$.
    \item[Yes:] There exists an $n$-qubit weight-$k$
          quantum state $\ket{\psi}$, such that
          $\Pr[\text{$\mathcal{C}(\ket{\psi})$ accepts}] \geq c$.
    \item[No:] For every $n$-qubit weight-$k$ quantum state $\ket{\psi}$,
          $\Pr[\text{$\mathcal{C}(\ket{\psi})$ accepts}] \leq s$.
 \end{description}
\end{definition}

\begin{definition}[{$\QW_{c,s}[t]$}]\label{def:QW}
  For $t\in\mathbb{N}$, the class $\QW_{c,s}[t]$ consists of all parameterized
  problems that are FPQT reducible to \problem{Weight-$k$ Weft-$t$ Depth-$d$
    Quantum Circuit Satisfiability$(c,s)$} for some constant depth $d \geq t$.
\end{definition}

Due to the constant depth requirement of weft-$t$ quantum circuits, it is not
clear if this class has the error reduction property.
These classes are most relevant when $c$ and $s$ have a polynomial gap, i.e.,
$c-s>1/\mathrm{poly}(n)$.
Based on this, we define the $\QW$-hierarchy as
\begin{definition}
  Define $\QW[t]$ as
  \begin{equation*}
    \QW[t] := \bigcup_{\mathclap{\substack{c,s\\ c-s>1/\mathrm{poly}(n)}}}
    \QW_{c,s}[t].
  \end{equation*}
\end{definition}

We have considered a slight variation for the definition of $\QW[t]$ as compared
to that in~\cite{bremner2022parameterized} where this class did not include a
possible dependence on $n$ in the completeness and soundness parameters.
We have chosen the present definition as we want to allow for the possibility of
a polynomial gap.
Central to our work is the weighted version of the local Hamiltonian problem
which we prove is in $\QW[1]$.
As is mentioned in the introduction, it was shown
in~\cite{bremner2022parameterized} this problem is in \XP (for a definition of
this class see~\cite{downey1999parameterized}), which is in stark contrast to
the \problem{Weight-$k$ Quantum Circuit Satisfiability} problem whose
slices are \BQP-hard and hence cannot be in \XP unless $\PT=\BQP$.
By proving that the weighted local Hamiltonian problem is in $\QW[1]$ we
demonstrate in this paper a likely separation between this problem and other
parameterized variants of \QMA-complete problems such as Quantum Circuit
Satisfiability under \FPQT{} reductions.

Define the weighted version of the local Hamiltonian
problem~\cite{bremner2022parameterized} as
\begin{definition}[\problem{Weight-$k$ $\ell$-Local
  Hamiltonian$(a,b)$}]\label{def:weight-LH}\hfill
  \begin{description}
    \item[Instance:] An $\ell$-local Hamiltonian $H\coloneqq\sum_iH_i$ on $n$
          qubits that comprises at most a polynomial in $n$ many terms
          $\{H_i\}$, which each act non-trivially on at most $\ell$ qubits and
          have operator norm $\norm{H_i}\leq 1$.
    \item[Parameter:] A natural number $k$.
    \item[Yes:] There exists an $n$-qubit weight-$k$ quantum state $\ket{\psi}$,
          such that $\matrixel{\psi}{H}{\psi} \leq a$.
    \item[No:] For every $n$-qubit weight-$k$ quantum state
          $\ket{\psi}$, $\matrixel{\psi}{H}{\psi} \geq b$.
  \end{description}
\end{definition}

\section{Weighted Local Hamiltonian is in $\QW[1]$}\label{sec:LH-qw1}

In this section we prove that the weighted version of the Local Hamiltonian
problem is in the class $\QW[1]$.
We state this as a theorem.

\begin{theorem}\label{thm:weighted-ham}
  Given $a,b$ such that $b-a>1/\mathrm{poly}(n)$, then \problem{Weight-$k$
    $\ell$-Local Hamiltonian$(a,b)$} is in $\QW_{c,s}[1]$ for
  some $c,s$ such that $c-s>1/\mathrm{poly}(n)$.
\end{theorem}

The proof of \cref{thm:weighted-ham} consists of a series of reductions.
In the first step, we reduce the weighted local Hamiltonian problem to a
weight-preserving quantum circuit satisfiability problem defined below.
This step is discussed in \cref{sec:weight-preserving-qcsat}.

\begin{definition}[Weight-$k$ Weight-Preserving Quantum Circuit
  Satisfiability$(c,s)$]\label{def:weightP-wQCSAT}\hfill
  \begin{description}
    \item[Instance:] A weight-preserving quantum circuit $\mathcal{C}$
          on $n$ witness qubits, $\poly(n)$ ancilla qubits with circuit
          size $\mathrm{poly}(n)$.
    \item[Parameter:] A natural number $k$.
    \item[Yes:] There exists an $n$-qubit weight-$k$ quantum
          state $\ket{\psi}$, such that
          \begin{equation*}
            \Pr[\text{$\mathcal{C}(\ket{\psi})$ accepts}] \geq c.
          \end{equation*}
    \item[No:] For every $n$-qubit weight-$k$ quantum state
          $\ket{\psi}$,
          \begin{equation*}
            \Pr[\text{$\mathcal{C}(\ket{\psi})$ accepts}] \leq s.
          \end{equation*}
  \end{description}
\end{definition}

\begin{remark}
  Note that when initializing the ancilla qubits, we can set at most $f(k)$ of
  them to $\ket{1}$, where $f$ is some computable function.
  This guarantees our reduction still contained in \FPQT.
  Also this problem doesn't require the circuit to be constant depth.
  We will design a constant depth circuit in the last step
  (in \cref{sec:in-qwone}).
\end{remark}

In the second step (\cref{sec:strong-error-reduction}), we prove that strong
completeness and soundness error reduction is also possible for the
weight-preserving circuits using the quantum singular value transformation.
This step is necessary for the reductions in the later steps.
In the third step we reduce the weight-preserving quantum circuit satisfiability
problem to instances of the Local Hamiltonian problem that are \emph{almost
  spatially sparse}.
This notion will be defined below in \cref{sec:spatially-sparse} of the proof of
\cref{thm:weighted-ham}.
Finally in the fourth step (\cref{sec:in-qwone}), we reduce the weighted almost
spatially sparse Hamiltonian to an instance of the weighted constant-depth,
weft-$1$, quantum circuit satisfiability problem.
Before proceeding to the proof of these reductions, we will prove some
preliminary results about weight-preserving quantum circuits first.

\subsection{Universality of Weight-Preserving Circuits}\label{sec:universal}

In this section, we will show how the classic proof of quantum universality
in~\cite{Barenco1995gates} can be adapted to show universality of
weight-preserving circuits.

\begin{definition}\label{def:wp-op}
  An operator $O$ acting on ${(\mathbb{C}^{2})}^{\otimes n}$ is
  weight-preserving if for any $k$ and any computational basis state $\ket{x}$
  of weight $k$, $O\ket{x}$ is a vector in ${(\mathbb{C}^{2})}^{\otimes n}$ of
  weight exactly $k$.
\end{definition}

\begin{definition}\label{def:circ-weight-preserving}
  A circuit $C$ is \emph{weight-preserving} if its corresponding unitary
  operator is weight-preserving.
\end{definition}

We also define the weight-preserving version of one-qubit gates.

\begin{definition}\label{def:U-weightpreserving}
  For any single qubit gate $U$, define a two-qubit gate
\begin{equation*}
  \hat{U} =
  \begin{pmatrix}
    1 & 0 & 0\\
    0 & U & 0\\
    0 & 0 & 1
  \end{pmatrix}.
\end{equation*}
\end{definition}
It is easy to check that $\hat{U}$ is always a weight-preserving gate.
Note that When $U = X$, $\hat{U}$ is the SWAP gate, this fact will be used
regularly below.

The Fredkin gate (control-SWAP gate) is another example of weight-preserving
gate.
We will also need in \cref{lem:weight-universal} the following weight-preserving
gate $E=\begin{pmatrix} 1 & 0\\ 0 & e^{i\delta} \end{pmatrix}$.
This phase gate is necessary for universality as otherwise we will not be able
to create relative phases between states such as $\ket{00}$ and $\ket{11}$.

\begin{definition}
  A set of weight-preserving gates is \emph{weight-universal} if they can
  (approximately) generate all weight-preserving unitary transformations.
\end{definition}

\begin{lemma}\label{lem:weight-universal}
  If a set of single-qubit gates $U_{1}, U_{2}, \ldots, U_{s}$ and CNOTs form a
  standard universal gate set, then
  $\hat{U}_{1}, \hat{U}_{2}, \ldots, \hat{U}_{s}$, Fredkin and $E$ gates form a
  weight-universal gate set when allowed two extra ancilla qubits in the state
  $\ket{01}$.
\end{lemma}

\begin{proof}
  We follow the steps of~\cite[Chapt.~4]{NC00}.
  In this proof the first step is to show that two-level unitary gates are
  universal and can generate any $d\times d$ unitary from the group $U(d)$.
  Recall that two-level unitaries are gates which only act on the subspace
  spanned by two computational basis state, for example for $d=3$ a two-level
  unitary could be
  
  \begin{equation*}
    \begin{pmatrix}
    a & 0 & b\\
    0 & 1 & 0\\
    c & 0 & d
  \end{pmatrix}.
  \end{equation*}
  
  The authors prove that $d\times d$ unitaries can be obtained using $d(d-1)/2$
  two level unitaries.
  In our case we simply need to recognize that this proof will hold in any
  chosen weight-$k$ subspace.
  Hence we can always use the same inductive steps as those
  in~\cite[Sec.~4.5.1]{NC00} where non-trivial unitaries are limited to this
  subspace.
  This requires at most $\binom{n}{k} \left(\binom{n}{k}-1\right) / 2$.
  
  Then, by following the proof in~\cite[Sec.~4.5.2]{NC00} it can be shown that
  if we can implement all $\hat{U}$ operators (where $U$ is a single qubit
  gate), $E$, and Fredkin, then we can implement any two-level unitary.

  Recall that in~\cite[Sec.~4.5.2]{NC00} the authors use the Gray code, which
  given two bitstrings generates a sequence of strings that differ by a single
  bit.
  That is the hamming weight changes by one in each step of the sequence.
  This sequence is used to generate a circuit of multiply-controlled
  single-qubit gates to define an arbitrary two-level unitary.

  In our case, we cannot use this construction as it is not weight-preserving.
  However, note that we have the Fredkin gate in our gate set, which allows
  controlled swaps, and also note that we are operating in a weight-preserving
  space.
  Hence, we only need a sequence of operations that controllably swap qubits in
  this space and then will ultimately perform $\hat{U}$ gate.
  Suppose we want to implement a two level operator in the subspace of
  $\ket{s}=\ket{10001}$ and $\ket{t}=\ket{11000}$.
  We can consider the following transformations $10001 \to 10100 \to 11000$.
  Essentially, we want to place $(k-1)$ of the $1$'s from $\ket{s}$ in the same
  positions of $(k-1)$ $1$'s in $\ket{t}$.
  The remaining non-swapped $1$ of $\ket{s}$ is placed in a position next to the
  remaining $1$ in $\ket{t}$, for instance in the previous example we performed
  the transformation $10001 \to 10100$ placing the last $1$ in the third
  position, next to the second position where the last $1$ of $\ket{t}$ is
  located.
  This can be implemented in the same way as in~\cite[Sec.~4.5.2]{NC00} with the
  difference that now we apply controlled $\mathrm{SWAP}$ operators controlled
  on the rest of the qubits, see \cref{fig:weightcode}.
  Finally the operator $\hat{U}$ acts on qubits $2$ and $3$ (corresponding to
  the second and third bits from left to right).
  This operator is controlled on the rest of the qubits and finally we revert
  the $\mathrm{SWAP}$ operations.
  For weight-$k$ states we will require at most $2k$ SWAP gates plus the
  controlled $\hat{U}$.
  
  \begin{figure}[ht]
    \centering
    \begin{tikzpicture}[node distance=1cm,
      gate/.style={draw, minimum size=16, fill=ChannelColor, inner sep=2},
      ctrl/.style={circle, fill=black, minimum size=4, inner sep=0},
      octrl/.style={circle, draw=black, fill=white, minimum size=4, inner sep=0},
      target/.style={circle, draw, minimum size=8, inner sep=0},
      measure/.style={draw, minimum height=14, minimum width=16, fill=gray},
      readout/.style={draw, minimum height=9, minimum width=12, fill=white},
      crossx/.style={path picture={\draw[inner sep=0pt]
          (path picture bounding box.south east) --
          (path picture bounding box.north west)
          (path picture bounding box.south west) --
          (path picture bounding box.north east);}},
      box/.style={draw, fill=ChannelColor, minimum height=1.8cm, minimum width=1.2cm}]

      \node (In0) {};
      \node[above of=In0] (In1) {};
      \node[above of=In1] (In2) {};
      \node[above of=In2] (In3) {};
      \node[above of=In3] (In4) {};
      \node (Out0) at (7cm,0) {};
      \node[above of=Out0] (Out1) {};
      \node[above of=Out1] (Out2) {};
      \node[above of=Out2] (Out3) {};
      \node[above of=Out3] (Out4) {};

      \draw (In0) -- (Out0)
      (In1) -- (Out1)
      (In2) -- (Out2)
      (In3) -- (Out3)
      (In4) -- (Out4);

      \node[crossx] (B1) at (1cm,0) {};
      \node[octrl] (B2) at (3.5cm,0) {};
      \node[crossx] (B3) at (6cm,0) {};

      \node[ctrl] (A1) at (1cm,4cm) {};
      \node[ctrl] (A2) at (3.5cm,4cm) {};
      \node[ctrl] (A3) at (6cm,4cm) {};

      \draw (A1) -- (B1.center)
      (A2) -- (B2)
      (A3) -- (B3.center);

      \node[octrl] at (1cm,1cm) {};
      \node[octrl] at (3.5cm,1cm) {};
      \node[octrl] at (6cm,1cm) {};

      \node[crossx] at (1cm,2cm) {};
      \node[crossx] at (6cm,2cm) {};

      \node[octrl] at (1cm,3cm) {};
      \node[octrl] at (6cm,3cm) {};

      \node[box] at (3.5cm,2.5cm) {$\hat{V}$};
    \end{tikzpicture}
    \caption{Circuit implementing a two-level unitary between states
      $\ket{s}=\ket{10001}$ and $\ket{t}=\ket{11000}$.
      The transformation represented by the controlled SWAP gates is
      $10001 \to 10100$.
      The controlled $\hat{V}$ gate implements the two-level transformation in
      the subspace spanned by $\ket{s}$ and $\ket{t}$.
      The black dots denote the control operations activated if the qubit is in
      the state $\ket{1}$ and white dots denote controls activated when the
      qubit is in state $\ket{0}$.
      The crosses indicate SWAP operations.}
    \label{fig:weightcode}
\end{figure}
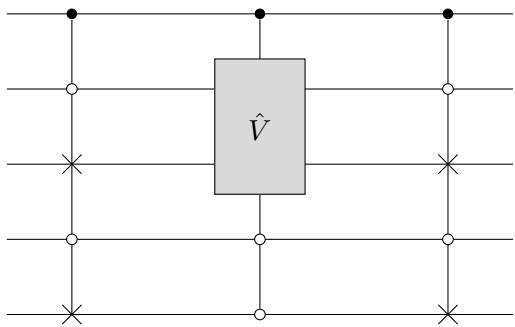

We now show that we can implement weight-$1$ two-qubit gates $\hat{V}$ with
multiple controls using only weight-$1$ two-qubit gates, the Fredkin, and $E$
gates.
We follow the technique employed in~\cite{Barenco1995gates} to prove this.
First, by Lemma 5.1 of~\cite{Barenco1995gates}, it's known that a controlled
version of $W\in \textrm{SU}(2)$ can be implemented by considering
$A,B,C\in \textrm{SU}(2)$ such that $ABC=I$ and $AXBXC=W$.
Directly employing the same decomposition in the case where $\hat{W}$ is a
weight-$1$ two-qubit gate, by noting that $\hat{A}\hat{B}\hat{C}=I$ and
$\hat{A}(\mathrm{SWAP})\hat{B}(\mathrm{SWAP})\hat{C}=\hat{W}$.
Note that in our case the CNOT gates become Fredkin gates.
To implement a single control version of $W\in \textrm{U}(2)$, a controlled
phase gate is included, in the weight-preserving case we use the gate $E$.

To construct a multiple controlled version of a unitary $\hat{W}$ with
$W\in \textrm{U}(2)$, consider the construction from Lemma 6.1
in~\cite{Barenco1995gates}.
We can create a weight-preserving version of this construction as in
\cref{fig:w-2control}, which includes two ancilla qubits set to $\ket{0}\ket{1}$
and requires finding $\hat{V}$ such that $\hat{V}^2=\hat{W}$.
This qubits can be reused for each gate we want to construct and thus only
increases the weight of all the qubits in $1$.
The intuition behind the circuit is that we use the $\ket{01}$ ancilla to decide
if we should apply the controlled $V^{\dag}$, since in the original construction
there are two CNOTs, we can replace them with SWAPs and the ancilla system.
When considering more control qubits, the construction generalizes in the same
way, by considering more Fredkin gates acting on ancillas instead of CNOTs.
Note that if we want the controls to be activated by $\ket{0}$ instead of
$\ket{1}$, we can simply introduce SWAPs in the ancilla system.
With these considerations we can implement any two-level unitary constructed
from circuits such as the one in \cref{fig:weightcode} using only weight-$1$ two
qubit gates, the Fredkin gate, and $E$.
If we want to use the discrete set $\hat{U}_1,\cdots,\hat{U}_s$ instead of all
weight-$1$ preserving two qubit gates, then the Solovay-Kitaev theorem applies
in this case and thus proves the result.
\end{proof}

\begin{remark}
  The proof above shows that to implement a two-level unitary over the
  weight-$k$ subspace requires $O(2^n)$ gates from our weight-universal gate
  set.
  This exponential comes mainly from the implementation we used for the
  controlled $\hat{W}$ gate.
  For our work in this paper, this exponential dependence is sufficient.
  We remark that a more efficient construction is possible, with caveat that it
  includes non-trivial operations outside the weight-$k$ subspace which might be
  of interest to some readers.
  In~\cite{Barenco1995gates} a more efficient construction is offered which
  scales like $O(n^2)$.
  We can adapt our proof to improve the scaling in the same way provided that we
  don't care how the two-level unitary acts outside the weight-$k$ subspace of
  dimension $2$.
  This improvement is obtained by noticing that circuits implementing two-level
  unitaries as in \cref{fig:weightcode} only require $k$ controls since we
  need to check the position of the 1's.
  This will imply that outside the weight-$k$ subspace the action of the unitary
  will be non-trivial, but if we only care about this subspace, then the
  dependence will be on $k$ rather than $n$ for implementing them.
  Even more improvements can be obtained using the techniques from Lemma 7.2 and
  Lemma 7.3 in~\cite{Barenco1995gates}.
\end{remark}

\begin{figure}[ht]
  \centering
  \begin{tikzpicture}[node distance=1cm and 0.5cm,
    gate/.style={draw, minimum size=16, fill=ChannelColor, inner sep=2},
    ctrl/.style={circle, fill=black, minimum size=4, inner sep=0},
    octrl/.style={circle, draw=black, fill=white, minimum size=4, inner sep=0},
    target/.style={circle, draw, minimum size=8, inner sep=0},
    measure/.style={draw, minimum height=14, minimum width=16, fill=gray},
    readout/.style={draw, minimum height=9, minimum width=12, fill=white},
    crossx/.style={path picture={\draw[inner sep=0pt]
        (path picture bounding box.south east) --
        (path picture bounding box.north west)
        (path picture bounding box.south west) --
        (path picture bounding box.north east);}},
    box/.style={draw, fill=ChannelColor, minimum height=1.8cm, minimum width=.8cm}]

    \begin{scope}
      \node (In0) at (0cm, 0cm) {};
      \node[above of=In0] (In1) {};
      \node[above of=In1] (In2) {};
      \node[above of=In2] (In3) {};
      \node (Out0) at (2cm, 0cm) {};
      \node[above of=Out0] (Out1) {};
      \node[above of=Out1] (Out2) {};
      \node[above of=Out2] (Out3) {};

      \draw (In0) -- (Out0)
      (In1) -- (Out1)
      (In2) -- (Out2)
      (In3) -- (Out3);

      \node[ctrl] at (1cm,2cm) {};
      \node[ctrl] (A) at (1cm,3cm) {};

      \node[box] (B) at (1cm,0.5cm) {$\hat{W}$};
      \draw (A) -- (B.north);

      \node at (3cm,2.5cm) {$\longrightarrow$};
    \end{scope}

    \begin{scope}[xshift=4.8cm]
      \node (In0) at (.5cm, 0cm) {};
      \node[above of=In0] (In1) {};
      \node[above of=In1] (In2) {};
      \node[above of=In2] (In3) {};
      \node[label=left:{$\ket{1}$}, above of=In3] (In4) {};
      \node[label=left:{$\ket{0}$}, above of=In4] (In5) {};
      \node (Out0) at (6.5cm, 0cm) {};
      \node[above of=Out0] (Out1) {};
      \node[above of=Out1] (Out2) {};
      \node[above of=Out2] (Out3) {};
      \node[above of=Out3] (Out4) {};
      \node[above of=Out4] (Out5) {};

      \draw (In0) -- (Out0)
      (In1) -- (Out1)
      (In2) -- (Out2)
      (In3) -- (Out3)
      (In4) -- (Out4)
      (In5) -- (Out5);

      \node[ctrl] (C1) at (3.5cm, 5cm) {};
      \node[box] (B1) at (3.5cm,.5cm) {$\hat{V}^{\dagger}$};
      \draw (C1) -- (B1.north);

      \node[ctrl] (C0) at (1.5cm, 2cm) {};
      \node[box] (B0) at (1.5cm,.5cm) {$\hat{V}$};
      \draw (C0) -- (B0.north);

      \node[ctrl] (C2) at (5.5cm, 3cm) {};
      \node[box] (B2) at (5.5cm,.5cm) {$\hat{V}$};
      \draw (C2) -- (B2.north);

      \begin{scope}[xshift=3cm]
        \node[ctrl] at (0cm,2cm) {};
        \node[crossx] at (0cm,4cm) {};
        \node[crossx] at (0cm,5cm) {};
        \draw (0cm,2cm) -- (0cm,5cm);
      \end{scope}

      \begin{scope}[xshift=4cm]
        \node[ctrl] at (0cm,2cm) {};
        \node[crossx] at (0cm,4cm) {};
        \node[crossx] at (0cm,5cm) {};
        \draw (0cm,2cm) -- (0cm,5cm);
      \end{scope}

      \begin{scope}[xshift=2.5cm]
        \node[ctrl] at (0cm,3cm) {};
        \node[crossx] at (0cm,4cm) {};
        \node[crossx] at (0cm,5cm) {};
        \draw (0cm,3cm) -- (0cm,5cm);
      \end{scope}

      \begin{scope}[xshift=4.5cm]
        \node[ctrl] at (0cm,3cm) {};
        \node[crossx] at (0cm,4cm) {};
        \node[crossx] at (0cm,5cm) {};
        \draw (0cm,3cm) -- (0cm,5cm);
      \end{scope}

    \end{scope}
  \end{tikzpicture}
  \caption{Circuit implementing a controlled version of $\hat{W}$ with two
    controls.
    This requires two ancillas initiated in the state $\ket{01}$ and can be
    reused in the construction of other gates.
    In this circuit $\hat{V}^2=\hat{W}$.}%
  \label{fig:w-2control}
\end{figure}
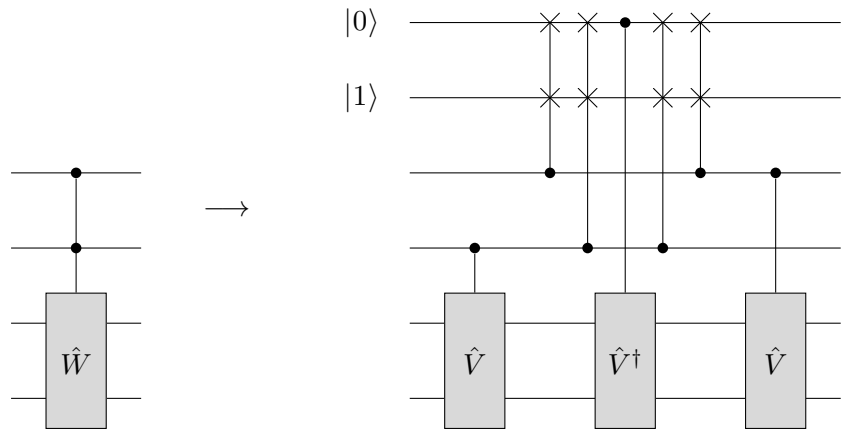

The following lemmas will be necessary for our proof of \cref{thm:weighted-ham}.

\begin{lemma}\label{lem:W-state-gen}
  Let $n = 2^{r}$ be an integer power of $2$.
  The W state
  \begin{equation*}
    \ket{W_n}=\frac{1}{\sqrt{n}}\left(\ket{10\cdots0}+\ket{01\cdots0}+
      \cdots+\ket{00\cdots1}\right)
  \end{equation*}
  of $n$ qubits can be computed from $\ket{0^{n-1}1}$ by a weight-preserving
  quantum circuit efficiently.
\end{lemma}

\begin{proof}
  We prove by induction on $r$ that there is such circuits $C_{n}$ such that
  $C_{n} \ket{0^{n}} = \ket{0^{n}}$ and $C_{n} \ket{0^{n-1}1} = \ket{W_{n}}$.
  First for $r=1$, the result follows by applying the gate
  \begin{equation}
    \label{eq:weight-H-gate}
    \begin{pmatrix}
      1 & 0 & \phantom{-}0 & 0\\
      0 & \frac{1}{\sqrt{2}} & \phantom{-}\frac{1}{\sqrt{2}} & 0\\
      0 & \frac{1}{\sqrt{2}} & -\frac{1}{\sqrt{2}} & 0\\
      0 & 0 & \phantom{-}0 & 1
    \end{pmatrix}.
  \end{equation}
  Assume the claim is proved for $n=2^{r-1}$ and we shall show the same for
  $n'=2^{r}$.
  Notice that
  \begin{equation*}
    \ket{W_{n'}} = \frac{1}{\sqrt{2}} \bigl( \ket{W_{n}} \ket{0^{n}} +
    \ket{0^{n}} \ket{W_{n}}\bigr),
  \end{equation*}
  which can be prepared by first apply the gate in \cref{eq:weight-H-gate} to
  the $n+1$ and the last qubit followed by two $C_{n-1}$ circuits acting on the
  first and second half of the qubits.
\end{proof}

\subsection{Weight-Preserving Quantum Circuit
  Satisfiability}\label{sec:weight-preserving-qcsat}

In this section, we construct a weight-preserving verification circuit from the
local Hamiltonian problem.
We emphasize that the Hamiltonian does \emph{not} need to be weight-preserving
and that the resulting circuit is not of constant depth yet.

\begin{lemma}\label{lem:wLH-to-wpQCSAT}
  Given a weight-$k$ $\ell$-local Hamiltonian problem $H = \sum_{j=1}^{m} H_{j}$ of
  $m$ terms on $n$ qubits and energy bounds $a$ and $b$ with gap
  $b-a > 1/\poly(n)$.
  Suppose also that $\norm{H_{j}} \le 1$ for all $j=1, 2, \ldots, m$.
  Then there is a weight-preserving circuit $W_{H}$ of $\poly(n)$ size on
  $n + M + k + 2$ qubits that accepts with probability
  \begin{equation*}
    1 - \frac{m + \bra{\psi} H \ket{\psi}}{2M}
  \end{equation*}
  where $\ket{\psi}$ is the input witness state and $M = 2^{\lceil\log_2m\rceil}$,
  the smallest integer power of $2$ larger than $m$.
\end{lemma}

\begin{proof}[Proof of \cref{lem:wLH-to-wpQCSAT}]
  We use $P^{(k)}_{m}$ to denote the projector onto the subspace of weight-$k$
  basis states of length $m$.
  By convention, If $k>m$ then $P^{(k)}_{m}$ is the zero operator.
  We first show how we can implement a weight-preserving unitary circuit that
  accepts with probability $\bra{\psi} (I-H_{j}) \ket{\psi}/2$.
  Assume for simplicity that the term $H_{j}$ acts on the first $\ell$ qubits
  and let $O = (I-H_{j})/2$ be a positive semi-definite operator.
  We are interested in the quantity $\bra{\psi}O\ket{\psi}$ and we claim the
  following identity
  \begin{equation*}
    \bra{\psi} O \otimes I_{n-l} \ket{\psi} =
    \sum_{w=0}^{l'} \bra{\psi} O^{(w)} \otimes P^{(k-w)}_{n-l} \ket{\psi}
  \end{equation*}
  for state $\ket{\psi}$ of weight $k$ where
  $O^{(w)} = P^{(w)}_{l} O P^{(w)}_{l},l'=\min(k,l)$.
  This follows by computing the matrix entries of $O\otimes I_{n-l}$ with
  indices $i,i'$ of weight $k$.
  Alternatively, one can see that

  \begin{align*}
    \matrixel{\psi}{O\otimes I}{\psi}
    & = \Bigbra{\psi} \left(\sum_{w=0}^{l'} P^{(w)}_{l} \otimes P^{(k-w)}_{n-l}
      \right) O \otimes I \left(\sum_{w'=0}^{l'}P^{(w')}_{l} \otimes
      P^{(k-w')}_{n-l} \right) \Bigket{\psi}\\
    & = \bigbra{\psi} \sum_{w=0}^{l'} P^{(w)}_{l} O P^{(w)}_{l} \otimes
      P^{(k-w)}_{n-l} \bigket{\psi}\\
    & = \bigbra{\psi} \sum_{w=0}^{l'} O^{(w)}\otimes P^{(k-w)}_{n-l} \bigket{\psi}
  \end{align*}

  Now we introduce two ancilla qubits starting in state $\ket{01}$.
  Then the following matrix
  \begin{equation*}
    U^{(w)} =
    \begin{pmatrix}
      I & 0 & 0 & 0\\
      0 & \sqrt{O^{(w)}} & \sqrt{I-O^{(w)}} & 0\\
      0 & \sqrt{I-O^{(w)}} & -\sqrt{O^{(w)}} & 0\\
      0 & 0 & 0 & I
    \end{pmatrix}
  \end{equation*}
  is unitary and weight-preserving.
  It is unitary as $U^{(w)} \bigl( U^{(w)} \bigr)^{\dagger} = I$ follows by
  direct calculations.
  The weight-preserving property follows from the weight-preserving property of
  $O^{(w)}$, and therefore also $\sqrt{O^{(w)}}$ and $\sqrt{I - O^{(w)}}$.
  The ancilla qubits in the state $\ket{01}$ are chosen such that
  $U^{(w)} \ket{\psi}\ket{01} = \sqrt{O^{(w)}}
  \ket{\psi}\ket{01} + \sqrt{I-O^{(w)}} \ket{\psi}\ket{10}$.
  We want to act with $U^{(w)}$ conditioned on the remaining $n-l$ qubits having
  weight $k-w$.
  We can do this by adding $k+1$ ancillas in the state $\ket{100\cdots0}$ and
  then act on this ancilla registers with controlled gates that perform a cyclic
  shift of the registers controlled by the original $n-l$ qubits.
  We define $S$ as the circular shift operator that act as
  $S\ket{i_1i_2\dots i_n}=\ket{i_ni_1\dots i_{n-1}}$.
  We can define the circuit $V_{\text{weight}}$ formally as

  \begin{align*}
    V_{\text{weight}}=\sum_{i=0}^{l'} P^{(k-i)}_{n-l}\otimes S^i.
  \end{align*}

  The circuit is drawn in \cref{fig:circ_localOp}.
  If the remaining $n-l$ have weight $k-i$, then the $k+1$ ancillas gets rotated
  from $\ket{10^k}$ to $\ket{0^i10^{k-i}}$.
  Consider the probability that we measure the first group of ancillary qubits
  in basis $\ket{01}$,
  \begin{align*}
  \norm{(\ket{01}\bra{01}\otimes I) U \ket{01}\ket{\psi} \ket{10^{l'}}}^2
  & = \norm{\ket{01}\otimes\left(\sum_{w=0}^{l'}\sqrt{O^{(w)}} \otimes
    P^{(k-w)}_{n-l}\ket{\psi}\otimes\ket{0^w10^{k-w}}\right)}^2\\
  & = \matrixel{\psi}{\sum_{w=0}^{l'} O^{(w)}\otimes P^{(k-w)}_{n-l}}{\psi}\\
  & = \bra{\psi}O\ket{\psi}.
  \end{align*}

  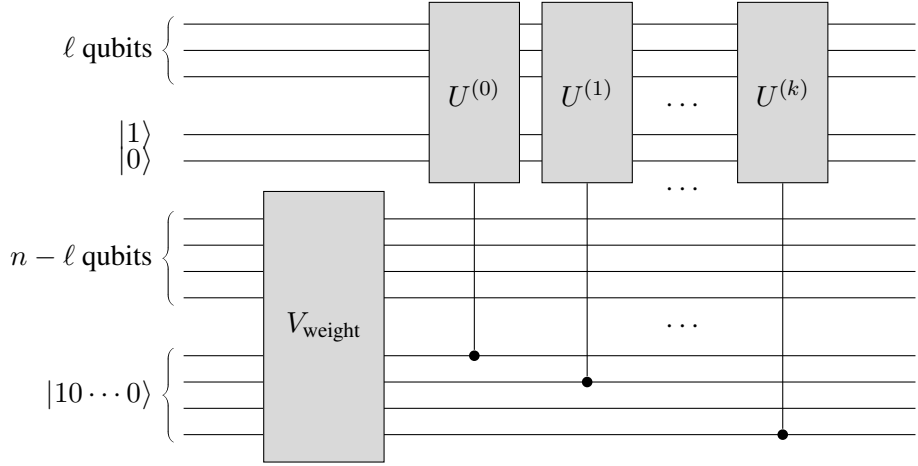
\begin{figure}[ht]
    \centering
    \begin{tikzpicture}[node distance=0.35cm,
      gate/.style={draw, minimum size=16, fill=ChannelColor, inner sep=2},
      ctrl/.style={circle, fill=black, minimum size=4, inner sep=0},
      octrl/.style={circle, draw=black, fill=white, minimum size=4, inner sep=0},
      target/.style={circle, draw, minimum size=8, inner sep=0},
      brace/.style={decorate, decoration = {calligraphic brace, amplitude=4}},
      measure/.style={draw, minimum height=14, minimum width=16, fill=gray},
      readout/.style={draw, minimum height=9, minimum width=12, fill=white},
      crossx/.style={path picture={\draw[inner sep=0pt]
          (path picture bounding box.south east) --
          (path picture bounding box.north west)
          (path picture bounding box.south west) --
          (path picture bounding box.north east);}},
      box/.style={draw, fill=ChannelColor, minimum height=2.4cm, minimum width=1.2cm},
      big/.style={draw, fill=ChannelColor, minimum height=3.6cm, minimum width=1.6cm},
      ]

      \node (In0) at (0,0) {};
      \node[above of=In0] (In1) {};
      \node[above of=In1] (In2) {};
      \node[above of=In2] (In3) {};
      \node[above=.5cm of In3] (In4) {};
      \node[above of=In4] (In5) {};
      \node[above of=In5] (In6) {};
      \node[above of=In6] (In7) {};
      \node[above=.5cm of In7] (In8) {};
      \node[above of=In8] (In9) {};
      \node[above=.5cm of In9] (InA) {};
      \node[above of=InA] (InB) {};
      \node[above of=InB] (InC) {};

      \node (Out0) at (10cm,0) {};
      \node[above of=Out0] (Out1) {};
      \node[above of=Out1] (Out2) {};
      \node[above of=Out2] (Out3) {};
      \node[above=.5cm of Out3] (Out4) {};
      \node[above of=Out4] (Out5) {};
      \node[above of=Out5] (Out6) {};
      \node[above of=Out6] (Out7) {};
      \node[above=.5cm of Out7] (Out8) {};
      \node[above of=Out8] (Out9) {};
      \node[above=.5cm of Out9] (OutA) {};
      \node[above of=OutA] (OutB) {};
      \node[above of=OutB] (OutC) {};

      \draw (In0) -- (Out0)
      (In1) -- (Out1)
      (In2) -- (Out2)
      (In3) -- (Out3)
      (In4) -- (Out4)
      (In5) -- (Out5)
      (In6) -- (Out6)
      (In7) -- (Out7)
      (In8) -- (Out8)
      (In9) -- (Out9)
      (InA) -- (OutA)
      (InB) -- (OutB)
      (InC) -- (OutC);

      \node[big] (big) at ([xshift=2cm]$(In3) !.5! (In4)$) {$V_{\text{weight}}$};

      \node[ctrl] (C0) at ([xshift=4cm]$(In3)$) {};
      \node[box] (B0) at ([xshift=4cm]$(In8) !.5! (InC)$) {$U^{(0)}$};
      \draw (C0) -- (B0.south);

      \node[ctrl] (C1) at ([xshift=5.5cm]$(In2)$) {};
      \node[box] (B1) at ([xshift=5.5cm]$(In8) !.5! (InC)$) {$U^{(1)}$};
      \draw (C1) -- (B1.south);

      \node at ([xshift=6.8cm]$(In7) !.5! (In8)$) {$\cdots$};
      \node at ([xshift=6.8cm]$(In3) !.5! (In4)$) {$\cdots$};
      \node at ([xshift=6.8cm]$(In9) !.5! (InA)$) {$\cdots$};

      \node[ctrl] (C2) at ([xshift=8.1cm]$(In0)$) {};
      \node[box] (B2) at ([xshift=8.1cm]$(In8) !.5! (InC)$) {$U^{(k)}$};
      \draw (C2) -- (B2.south);

      \draw[brace] ([yshift=-.1cm]$(In0)$) -- ([yshift=.1cm]$(In3)$)
      node[pos=0.5, label=left:{$\ket{10\cdots 0}$}] {};

      \draw[brace] ([yshift=-.1cm]$(In4)$) -- ([yshift=.1cm]$(In7)$)
      node[pos=0.5, label=left:{$n-\ell$ qubits}] {};

      \draw[brace] ([yshift=-.1cm]$(InA)$) -- ([yshift=.1cm]$(InC)$)
      node[pos=0.5, label=left:{$\ell$ qubits}] {};

      \node[label=left:{$\ket{0}$}] at (In8) {};
      \node[label=left:{$\ket{1}$}] at (In9) {};
    \end{tikzpicture}
    \caption{Circuit implementing the observable $O=(I-H_j)/2$ described in the
      text.
      The unitary $V_{\text{weight}}$ writes the weight of the $n-l$ qubits on the
      counting registry $\ket{10\cdots0}$.
      The circuit acts on the $\ell$ qubits (and the pair of ancillas) depending on
      this weight.}
    \label{fig:circ_localOp}
  \end{figure}

  We are now ready to construct the weight-preserving circuit for the local
  Hamiltonian $H$.
  It consists of two registers of qubits.
  The first is the term selection register of $M = 2^{\lceil\log_2m\rceil}$
  qubits.
  The second register contains $n$ qubits representing the witness state to the
  Hamiltonian problem.
  The circuit starts with the preparation of the $M$-qubit $\ket{W}$ state in
  the term selection register.
  For all $j=1, 2, \ldots, m$ and conditioned on the $j$-th qubit in the term
  selection register being in state $\ket{1}$, we perform the network of SWAP
  gates that moves the qubits that $H_{j}$ acts on to the first $\ell$ qubits,
  apply the weight-preserving energy measurement circuit for $O=(I-H_{j})/2$ as
  described above, note that the measurement performed depends on the chosen $j$
  as well.
  For all $j=m+1, \ldots, M$, the circuit accepts immediately.

  It is easy to check that all gates used in the circuit are weight-preserving
  and the circuit accepts with probability
  \begin{equation*}
    \frac{M-m}{M} + \sum_{j=1}^{m}  \frac{1 - \bra{\psi}  H_{j} \ket{\psi}}{2M} =
    1- \frac{m + \bra{\psi} H \ket{\psi}}{2M}.
  \end{equation*}

\end{proof}

\subsection{Weight-Preserving Marriott-Watrous Amplification}\label{sec:strong-error-reduction}

In this section, we prove that it is possible to amplify the completeness and
soundness gap for weight-preserving verification circuits with one copy of the
witness state.
For standard $\QMA$ verifiers, this is known as the strong
completeness-soundness gap amplification first established by Marriott and
Watrous in~\cite{marriott2005quantum}, where the construction iteratively measures the
post-measurement states of the verifier circuit in some structured way and makes
the final decision by performing a counting procedure on the measurement
outcomes.
This standard construction does not fit will in the weight-preserving scenario
as it is difficult to encode polynomially many measurement outcome bits in a
Hilbert space of dimension roughly $n^{k}$.

For this reason, we would use the fast $\QMA$
reduction~\cite{NWZ09,gilyen2019quantum}.
We employ a version inspired by the quantum singular value transformation (QSVT)
algorithm in the following~\cite{gilyen2019quantum} to amplify the error gap of
the verification circuit in a weight preserving manner.

\begin{theorem}\label{thm:FastQMA}
  Given a verifier circuit $V$ for a language $L\in\QMA$ with acceptance
  probability thresholds $(a,b)$, we can construct a new verifier circuit $V'$
  with threshold $a'=\epsilon, b'=1-\epsilon$ with one extra ancillary qubit,
  and $m = O \left(\frac{1}{\max{[ \sqrt{b} - \sqrt{a}, \sqrt{1-a} -
        \sqrt{1-b}]}} \log(\frac{1}{\epsilon}) \right)$
  calls to $V$ and $V^{\dagger}$ as in \cref{fig:SVT}.
\end{theorem}

By careful examination of the new circuit constructed
in~\cite{gilyen2019quantum}, we can show that the circuit could be implemented
in a weight preserving manner, giving us the following corollary:

\begin{corollary}\label{col:FastwpMW}
  Given an instance circuit $\mathcal{C}$ of weight-$k$ weight preserving
  quantum circuit with completeness and soundness $c,s, (c-s>1/\poly(n))$, we
  can construct a new weight preserving circuit $\mathcal{C}'$ with threshold
  $c'=1-\epsilon, s'=\epsilon$, by making $\poly(n)\log(1/\epsilon)$ calls to
  the circuit $\mathcal{C}$.
\end{corollary}

Now we explicitly write out the circuit in the previous construction.
If we assume
$V\ket{\psi}\ket{1^{f(k)}0^p}=\alpha\ket{1}\ket{\varphi_1}+\beta\ket{0}\ket{\varphi_0}$,
let $\Phi\in \mathbb{R}^{2m}$, define the following circuit $U_{\Phi}$:
\begin{align*}
  U_{\Phi}=\prod_{j=1}^n\left(e^{i\phi_{2j-1}(2\Pi-I)}V^{\dagger}e^{i\phi_{2j}(2\tilde{\Pi}-I)}V\right),
\end{align*}
where the $\Pi=I\otimes\ketbra{1^{f(k)}0^p}$ is the projector that checks the
ancillary qubits are correctly initialized, and
$\tilde{\Pi}=\ketbra{1}\otimes I$ is the accepting projector on the output qubit
of $V$.

It is shown in~\cite{gilyen2019quantum} that there exists some
$\Phi\in \mathbb{R}^{2m}$, where $m$ is set as in \cref{thm:FastQMA}, such that
\begin{align*}
  ||(\bra{+}\otimes\Pi)(\ket{0}\bra{0}\otimes
  U_{\Phi}+\ket{1}\bra{1}\otimes U_{-\Phi})
  (\ket{+}\otimes\ket{\psi})||^2
  & \geq 1-\epsilon,
  &\text{ if }||\tilde{\Pi}V\ket{\psi}||^2\geq b;\\
  ||(\bra{+}\otimes\Pi)(\ket{0}\bra{0}\otimes U_{\Phi}+
  \ket{1}\bra{1}\otimes U_{-\Phi})(\ket{+}\otimes\ket{\psi})||^2
  &\leq \epsilon,
  &\text{ if }||\tilde{\Pi}V\ket{\psi}||^2\leq a;
\end{align*}

To implement the controlled $U_{\Phi}$ in the previous formula, we only need to
implement the gates $\sum_b\ket{b}\bra{b}\otimes e^{i(-1)^b\phi(2\Pi-I)}$ as in
\cref{fig:SVT}.

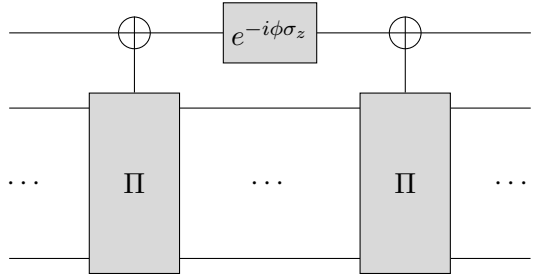
\begin{figure}[ht]
  \centering  
  \begin{tikzpicture}[scale=1.2, node distance=1cm,
    gate/.style={draw, minimum height=0.8cm, fill=ChannelColor, inner sep=3},
    ctrl/.style={circle, fill=black, minimum size=4, inner sep=0},
    octrl/.style={circle, draw=black, fill=white, minimum size=4, inner sep=0},
    target/.style={circle, draw, minimum size=12, inner sep=0},
    brace/.style={decorate, decoration = {calligraphic brace, amplitude=3}},
    measure/.style={draw, minimum height=14, minimum width=16, fill=gray},
    readout/.style={draw, minimum height=9, minimum width=12, fill=white},
    crossx/.style={path picture={\draw[inner sep=0pt]
        (path picture bounding box.south east) --
        (path picture bounding box.north west)
        (path picture bounding box.south west) --
        (path picture bounding box.north east);}},
    box/.style={draw, fill=ChannelColor, minimum height=2.4cm, minimum width=1.2cm}]

    \begin{scope}
      \node (In0) at (0,0) {};
      \node[above of=In0] (In1) {};
      \node[above of=In1] (In2) {};
      \node[above of=In2] (In3) {};

      \node (Out0) at (6cm,0) {};
      \node[above of=Out0] (Out1) {};
      \node[above of=Out1] (Out2) {};
      \node[above of=Out2] (Out3) {};

      \draw (In0) -- (Out0)
      (In2) -- (Out2)
      (In3) -- (Out3);

      \node[target] (T) at ([xshift=1.5cm]$(In3)$) {};
      \node[box] (B) at ([xshift=1.5cm]$(In0) !.5! (In2)$) {$\Pi$};
      \draw (T.north) -- (B.north);

      \node[gate] at ($(In3) !.5! (Out3)$) {$e^{-i\phi\sigma_{z}}$} {};

      \node[target] (T) at ([xshift=4.5cm]$(In3)$) {};
      \node[box] (B) at ([xshift=4.5cm]$(In0) !.5! (In2)$) {$\Pi$};
      \draw (T.north) -- (B.north);

      \node at ([xshift=.3cm]$(In1)$) {$\cdots$};
      \node at ([xshift=3cm]$(In1)$) {$\cdots$};
      \node at ([xshift=5.7cm]$(In1)$) {$\cdots$};

    \end{scope}
  \end{tikzpicture}
  \caption{Implementing $\sum_b\ket{b}\bra{b}\otimes e^{i(-1)^b\phi(2\Pi-I)}$.}\label{fig:SVT}
\end{figure}

The $\mathrm{C_{\Pi}NOT}$ gate is defined as $\Pi\otimes X+(I-\Pi)\otimes I$.
In our weight preserving reduction, we replace the circuit $V$ with our weight
preserving instance $\mathcal{C}$, and encode the ancillary qubit in the
$\{\ket{01},\ket{10}\}$ space as before, replacing all operations on the ancilla
with their weight preserving counterpart.
In the end, we measure the circuit with
$\left(\frac{(\ket{01}+\ket{10})(\bra{01}+\bra{10})}{2}\otimes\Pi\right)$, and
accept if the output is 1.
The projector $\Pi$ could be implemented by counting the weight of first $f(k)$
qubits and rest of qubits using the shifting trick.

We complete the proof by examining that each gate in the new constructed circuit
is weight preserving.

\subsection{Spatially Sparse Weighted Local Hamiltonian}\label{sec:spatially-sparse}

We now show that any weight-preserving circuit $W$ with $R$ gates acting on $n$ qubits and with a
weight-$k$ witness state can be transformed to a weight-$(2k+1)$ local
Hamiltonian problem that is \emph{almost spatially sparse} (defined below).
The almost spatially sparsity will be used in the end to prove that the
problem is in $\QW[1]$.

\begin{definition}[Spatially Sparse Local Hamiltonian]
  A local Hamiltonian problem is \emph{spatially sparse} if each qubit is only
  acted by $O(1)$ Hamiltonians.
\end{definition}

\begin{definition}[Almost Spatially Sparse Local Hamiltonian]
  A local Hamiltonian problem is \emph{almost} spatially sparse with respect to
  a register of qubits if the Hamiltonian becomes spatially sparse if we remove
  all terms acting only on qubits in this register.
\end{definition}

The spatially sparse local Hamiltonian is proven to be $\QMA$ complete
in~\cite{OT08}, their key lemma is stated as follows:

\begin{lemma}\label{lem:OTconstruct}
  Given a verifier circuit $V_x$ for a language $L\in \QMA$, there exists a
  spatially sparse local Hamiltonian $H=\sum_i H_i$ and $T=\poly(n)$ that
  satisfies the following conditions:
  \begin{itemize}
    \item If $V_x$ accepts some state $\ket{\xi}$ with probability $1-\epsilon$,
          there exists state $\ket{\psi}$ that
          $\bra{\psi}H\ket{\psi}\leq \frac{\epsilon}{T+1}$.
    \item If $V_x$ accepts any state $\ket{\xi}$ with probability no larger than
          $\epsilon$, then all eigenvalues of $H$ is larger than
          $\frac{c(1-\sqrt{\epsilon}-\epsilon)}{T^3}$, where $c$ is some
          constant.
 \end{itemize}
\end{lemma}

We closely follow the construction in~\cite{OT08} to prove our weight preserving
variant of \cref{lem:OTconstruct}.
We first transform the original verification circuit $V_x$ to an equivalent
circuit $U_{\text{sp}}$ on a grid, such that each qubit on the grid is only
acted upon by a constant number of gates.
Then we apply a modified version of Kitaev's circuit-to-Hamiltonian construction
to obtain our Hamiltonian.

Assume $V_x=U_R\dots U_2U_1$ acts on $n$ qubits, where $U_i$ are local gates
from a universal gate set.
We introduce a grid with $R+1$ layers, each consists of $n$ qubits.
Intuitively, we want to use the $i$th qubit in layer $j$ in $U_{\text{sp}}$ to
simulate the state of $i$th qubit at time step $j$ in $V_x$.

To simulate $V_x$, we initialize the qubits corresponding to the input qubits of
$V_x$ in the first layer as the witness state $\ket{\psi}$ for $V_x$, and rest
of qubits in the first layer as $V_x$ initial work space.
For the qubits in other layers, we initialize them as $\ket{0}$.
On the $i$th layer, we perform the nontrivial gate $U_i$ on the corresponding
qubits, then perform SWAP gates on the qubits in the same column between layer
$i$ and $i+1$.
In the $R+1$ layer, we perform the measurement on the output qubit.
It is easy to verify that the circuit $U_{\text{sp}}$ simulates $V_x$
faithfully.

Moreover, if our $V_x$ is a weight preserving circuit, since every gate $U_i$
and SWAP gate are weight preserving, we can see that $U_{\text{sp}}$ is also
weight preserving.
The order of actual computational gates, and SWAP gates are applied in the same
order as specified in~\cite{OT08}.

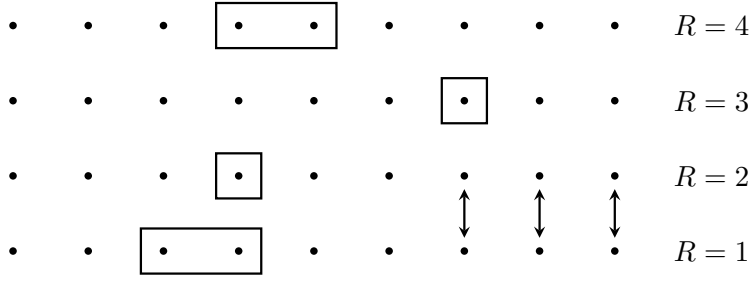
\begin{figure}
  \centering
  \begin{tikzpicture}[line width=.3mm, draw=black, >=stealth,
    shorten >=5pt, shorten <=5pt, auto,
    dot/.style={circle, fill=black, minimum width=1mm, inner sep=0},
    box/.style={rectangle, draw=black, minimum width=16mm, minimum height=6mm},
    square/.style={rectangle, draw=black, minimum height=6mm, minimum width=6mm}]

    \foreach \x in {0,1,...,8}{
      \foreach \y in {0,1,...,3} {
        \node[dot] at (\x,\y) {};
      }
    }

    \foreach \y in {1,2,...,4} {
      \node[label=right:{$R=\y$}] at (8.5,\y-1) {};
    }

    \node[box] at (2.5,0) {};
    \node[square] at (3,1) {};
    \node[square] at (6,2) {};
    \node[box] at (3.5,3) {};

    \draw[<->] (8,0) -- (8,1);
    \draw[<->] (7,0) -- (7,1);
    \draw[<->] (6,0) -- (6,1);

  \end{tikzpicture}
  \caption{Reproduced Figure 1 of~\cite{OT08}.
    Each row of the qubits has the same number as the starting circuit.
    The number of rows is one more than the number of gates in the starting
    circuit.
    The $R$-th gate is performed on the $R$-th row and then all qubuts are
    swappwd with those in the $(R+1$)-th row.
    This lazy simulation of the circuit will ensure that each qubit is acted on
    by a gate at most three times.
  }
  \label{fig:OT}
\end{figure}

Now we follow~\cite{OT08} and consider a variant of Kitaev's
circuit-to-Hamiltonian construction.
We denote $Q_{\text{in}}$ the set of qubits which correspond to the witness,
$Q_{\text{out}}$ the output qubit and $C=C_1,\cdots,C_T$ the clock registers.
The necessary change we make here is that we shall use an indicator clock
instead of the unary clock for maintaining the weight property.
Let $U_{\text{sp}}=W_T\dots W_2W_1$, the clock register will have $T+1$ qubits
and valid clock basis states have the form $\ket{0^{t-1}10^{T-t}}_{C}$ for
$t=1, 2, \ldots, T$.
Analogous to Kitaev's reduction, our legal history state for the circuit
$U_{\text{sp}}$ is
\begin{align*}
  \ket{\phi}=\frac{1}{\sqrt{T+1}}\sum_{t=0}^T
  \ket{0^{t-1}10^{T-t}}_{C}\otimes\ket{\xi_t},
\end{align*}
where
$\ket{\xi_t}=W_t\ket{\xi_{t-1}}, \ket{\xi_0} = \ket{\psi}\otimes\ket{1^{f(k)}0^m}$.
Since each $W_t$ is weight preserving, the initial state $\ket{\xi_0}$ for
$U_{\text{sp}}$ has weight $k+f(k)=k'$.

First, we recall the Hamiltonian construction in~\cite{OT08}, by replacing the
clock checking term for unary clock to indicator clock, we have the following
construction:
\begin{align*}
  H'_{\text{in}} & = \sum_{q\not\in Q_{\text{in}}}
                  \ketbra{\bar{i}_q}_{q} \otimes \ketbra{1}{1}_{C_{t_{q}-1}},\\
  H'_{\text{out}} & = \ketbra{0}_{Q_{\text{out}}} \otimes \ketbra{1}{1}_{C_{T}},\\
  H'_{\text{prop}} & = \sum_{t=1}^{T} H'_{\text{prop},t},
\end{align*}
and 
\begin{equation*}
  H'_{\text{prop},t} = (\ketbra{10}{10} + \ketbra{01}{01})_{C_{t,t+1}}
  - W_{t} \otimes \ketbra{01}{10}_{C_{t,t+1}} -
  W_{t} ^{\dagger} \otimes \ketbra{10}{01}_{C_{t,t+1}},
\end{equation*}
where $t_{q}$ stands for the earliest time step when qubit $q$ is actually used,
$\bar{i}_q$ is the inverse value in which ancillary qubit $q$ should be
initialized (so if qubit $q$ should be initialized in state $\ket{0}$, then
$\bar{i}_q=1$).
Here we omitted the clock checking term, since it will be reconstructed in our
final construction.

In the next step, we will perform the following isometry $\mathcal{U}$ on the
state registers of our Hamiltonian: for each qubit $q$, we will duplicate it in
the computational basis:
$\mathcal{U}\ket{0}_q=\ket{00}_{I_q},\mathcal{U}\ket{1}_q=\ket{11}_{I_q}$, where
$I_q$ are two qubits indicating the original qubit $q$, and
$I_q\cap I_{q'}=\emptyset$ for $q\neq q'$.

Thus our new Hamiltonian $H$ could be constructed by conjugating $\mathcal{U}$
over the previous construction $H=\mathcal{U}H'\mathcal{U}^{\dagger}$.
Our final local Hamiltonian will have the form
$H = H_{\text{in}} + H_{\text{out}} + H_{\text{prop}}+H_{\text{clock}}+H_{\text{state}}$
where
\begin{align*}
  H_{\text{in}}
  & = \sum_{q\not\in Q_{\text{in}}} \ketbra{\bar{i}_q\bar{i}_q
    }{\bar{i}_q\bar{i}_q}_{I_q} \otimes \ketbra{1}{1}_{C_{t_{q}-1}},\\
  H_{\text{out}}
  & = \ketbra{00}{00}_{I_{Q_{\text{out}}}} \otimes \ketbra{1}{1}_{C_{T}},\\
  H_{\text{clock}}
  & = \sum_{t<{t'}}\ket{11}\bra{11}_{C_{t,t'}}\\
  H_{\text{state}}
  & =\sum_q\ket{01}\bra{01}_{I_q}+\ket{10}\bra{10}_{I_q},\\
  H_{\text{prop}}
  & = \sum_{t=1}^{T} H_{\text{prop},t},
\end{align*}
and
\begin{equation*}
  H_{\text{prop},t} = (\ketbra{10}{10} + \ketbra{01}{01})_{C_{t,t+1}}
  - W'_{t} \otimes \ketbra{01}{10}_{C_{t,t+1}} -
  (W'_{t}) ^{\dagger} \otimes \ketbra{10}{01}_{C_{t,t+1}},
\end{equation*}
where
$W'_t=\mathcal{U}|_{Q_{W_t}}W_t\mathcal{U}|_{Q_{W_t}}^\dagger\otimes I_{2n-2|Q_{W_t}|}$,
$Q_{W_t}$ for the qubits that $W_t$ acts on.
In our new construction, our history state could be defined as
$\ket{\phi'}=(I\otimes\mathcal{U})\ket{\phi}$.
We can observe that $\mathcal{U}$ doubles the weight on the state registers, the
weight of our new witness state is $2k'+1$.

The difference between our construction and Oliveira-Terhal~\cite{OT08} is in
the clock design and checking terms.
We use the indicator clock and it is easy to see $H_{\text{clock}}$ and
$H_{\text{state}}$ are $2$-local Hamiltonians.
$H_{\text{state}}$ guarantees the two mapped qubits in $I_q$ always have the
same value, thus all legal witness should have even weight on the state
registers.
Since we require the weight of witness state to be odd, the clock registers must
have non-zero weight, and $H_{\text{clock}}$ guarantees the only valid clock
states are the indicator states $\ket{0^{t-1}10^{T-t}}_{C}$.

For the completeness part, observe that if original $V$ accepts $\ket{\psi}$
with probability $1-\epsilon$, the history state $\ket{\phi}$ would be projected
to 0 for all hamiltonian terms but $H_{\text{out}}$.
Since $U_{\text{sp}}$ simulates $V$ faithfully, we obtain that
$\bra{\phi}H_{\text{out}}\ket{\phi}\leq\frac{\epsilon}{T+1}$.

For the soundness part, observe that $H$ preserves the subspace of legal history
states
$\mathcal{S}=\{\ket{\phi} \colon H_{\text{clock}}\ket{\phi}=H_{\text{state}}\ket{\phi}=0\}$,
thus we can discuss the eigenvalue of $H$ on $\mathcal{S}$ and
$\mathcal{S}^{\perp}$ separately.
Since any eigenvector in $\mathcal{S}^{\perp}$ has eigenvalue at least 1, we can
focus on $H|_{\mathcal{S}}$.
Define $H'=H'_{\text{in}} + H'_{\text{out}} + H'_{\text{prop}}$, we have that
$\mathcal{U}H'\mathcal{U}^\dagger|_{\cal S}=H|_{\cal S}$.
In~\cite{OT08}, they performed analysis of eigenvalue on
$H'|_{\cal U^{\dagger}SU}$, which is isometric to
$\mathcal{U}H'\mathcal{U}^\dagger|_{\cal S}$, thus we obtain the same eigenvalue
lower bound $\frac{c(1-\sqrt{\epsilon}-\epsilon)}{T^3}$.

The resulting Hamiltonian in our reduction is not spatially sparse as
in~\cite{OT08} because the clock checking Hamiltonian $S_{\text{clock}}$ is not
sparse.
Excluding the clock checking terms, however, all other terms are spatially
sparse.
Therefore, this Hamiltonian is \emph{almost} spatially sparse with respect to
the clock register.
Note that if we use \cref{lem:weight-universal} and a finite gate set, the types
of resulting Hamiltonian terms will also be finite.
We conclude with the following corollary:

\begin{corollary}\label{col:wpOT}
  Given a weight-$k$ weight-preserving quantum circuit satisfiability instance
  $\mathcal{C}$ with parameter $(\epsilon,1-\epsilon)$, we can construct a
  weight-$2k'+1$ almost spatially sparse local hamiltonian instance with energy
  thresholds
  $a = \frac{\epsilon}{T+1}, b = \frac{c(1-\sqrt{\epsilon}-\epsilon)}{T^3}$.
  Furthermore, if we assume $\cal C$ acts on $n$ qubits, we have
  $T\leq 3n(|\mathcal{C}|+1)$, the resulting Hamiltonian would act on
  $2n(|\mathcal{C}|+1)+T+1$ qubits, and $k'=f(k)$ for some computable function
  $f$.
\end{corollary}

\subsection{QW[1] Verification for Almost Spatially Sparse Hamiltonian Problems}\label{sec:in-qwone}

We are now ready to show that the almost spatially sparse Hamiltonian problem we
end up with in the last subsection is in $\QW[1]$.
We shall design the constant depth circuit verifying the Hamiltonian problem
using a combination of two techniques described in the following.

First, we show how to check the spatially sparse terms in constant depth.
To do so, we color the terms using constant number of different colors so that
all terms having the same color act on different sets of qubits, a condition
that leads to constant-depth energy measurements of many Hamiltonian terms in
parallel.
This is easy to do by observing the structure of the terms in $H_{\text{in}}$,
$H_{\text{out}}$, $H_{\text{prop}}$, and \cref{fig:OT}.

Second, for the checking of the indicator clock format, we can simply measure
the clock register and perform classical $\W[1]$ computation to check the
result.
Thanks to the simplification of the clock checking term using the weight
constraint, it suffices to check that there are no two $1$'s in the measurement
outcome of the clock register.
This can be done in $\W[1]$, and therefore simulated by a constant depth quantum
circuit with one big AND gate.

We need the following lemma to relate parallel measurements and Hamiltonian sum
later on.

\begin{lemma}\label{lem:sum-vs-prod}
  Let $M_{1}, M_{2}, \ldots, M_{m}$ be $m$ commuting operators satisfying
  $0 \le M_{j} \le I$, then we have
  \begin{equation*}
    I - \sum_{j=1}^{m} M_{j} \le \prod_{j=1}^{m} (I - M_{j})
    \le I - \frac{1}{m} \sum_{j=1}^{m} M_{j}.
  \end{equation*}
\end{lemma}

\begin{proof}
  By the commutativity of the $m$ operators and the spectral decomposition
  theorem, this problem reduces to the scalar case.
  For real numbers $x_{j} \in [0,1]$ where $j=1, 2, \ldots, m$,
  \begin{equation*}
    1 - \sum_{j=1}^{m} x_{j} \le \prod_{j=1}^{m} (1-x_{j})
  \end{equation*}
  follows from a simple induction on $m$ and
  \begin{equation*}
    1 - \frac{1}{m} \sum_{j=1}^{m} x_{j} \ge \prod_{j=1}^{m} (1-x_{j})
  \end{equation*}
  follows from the geometric and arithmetic mean inequality
  \begin{equation*}
    \frac{\sum_{j} (1-x_{j})}{m} \ge {\Bigl( \prod_{j}(1-x_{j}) \Bigr)}^{1/m}
    \ge \prod_{j} (1-x_{j}).
  \end{equation*}
\end{proof}

\begin{lemma}\label{lem:almost-sparse-to-circ}
  Let $H = \sum_{j} H_{j}$ be a local Hamiltonian problem that acts on $n$
  qubits.
  The energy thresholds $a$ and $b$ for the problem satisfies
  $b/n^{2} - a \ge 1/\poly(n)$.
  Suppose that Hamiltonian $H$ is almost spatially sparse with respect to a
  clock register of $n_{\text{clock}}$ qubits and that each term $H_{j}$ in the
  Hamiltonian is a projector.
  That is, except clock checking terms $\ket{11}\bra{11}_{C_{t,t'}}$ acting on qubits
  $C_t$ and $C_{t'}$ in the clock register, all other Hamiltonian terms in $H$ are
  spatially sparse.
  Then, there is a $\QW[1]$ verification circuit $V$ and $c,s \in \real$
  satisfying $c-s\ge 1/\poly(n)$ such that if the ground state energy of $H$ is
  at most $a$, $V$ accepts with probability $c$ while if the ground state energy
  of $H$ is at least $b$, $V$ accepts with probability $s$.
  Furthermore, $V$ can be chosen so that the big gate is a classical AND gate
  and it is the last gate in $V$.
\end{lemma}

\begin{proof}
  As the Hamiltonian is almost spatially sparse, it is possible to color the
  terms using $n_{\text{color}} + 1$ colors where $n_{\text{color}}$ is a
  constant.
  We use $G^{(h)}$ to denote the set of terms of color $h$.
  For the first $n_{\text{color}}$ sets $G^{(h)}$ where
  $h=0, 1, \ldots, n_{\text{color}}-1$, the terms $H^{(h)}_{j}$ in the color
  group
  \begin{equation*}
    G^{(h)} = \bigl\{ H^{(h)}_{j} \mid j=1,2,\ldots, m_{h} \bigr\}
  \end{equation*}
  acts on different qubits for all $j$.
  Here, $m_{h}$ is the number of terms in group $G^{(h)}$.
  For the last group $G^{(n_{\text{color}})}$, the terms are
  $H^{(n_{\text{color}})}_{j} = \ket{11}\bra{11}_{C_{t,t'}}$ acting on all pairs of
  qubits $C_t,C_{t'}$ in the clock register.
  The number of terms in this group is $m_{n_{\text{color}}}$.
  Define
  \begin{equation*}
    m_{\max} = \max\, \{m_{i} \mid i=0,1,\ldots, n_{\text{color}}\}.
  \end{equation*}
  For each $h=0, 1, \ldots, n_{\text{color}}-1$, the size $m_{h}$ is at most $n$
  as the terms in $G^{(h)}$ all act on different qubits.
  For $h=n_{\text{color}}$, $m_{h}$ is at most $n^{2}$ as
  $n_{\text{clock}} \le n$ and the terms run over a pair of clock qubits.
  This implies that $m_{\max} \le n^{2}$.

  We now present the $\QW[1]$ verification circuit $V$ as follows.
  \begin{enumerate}
    \item First the circuit samples a random integer
          $h\in \{0, 1, \ldots, n_{\text{color}}\}$.
    \item Conditioned on $h$ the circuit checks all the terms in the group
          $G^{(h)}$.
          In particular,
    \begin{enumerate}
      \item If $h < n_{\text{color}}$, the circuit performs measurements
      \begin{equation*}
        \{M^{(h)}_{j, 1} = I - H^{(h)}_{j}, M^{(h)}_{j, 0} = H^{(h)}_{j}\},
      \end{equation*}
      for all $j=1, 2, \ldots, m_{h}$.
      The circuit outputs the AND of all measurement outcomes.
      \item If $h = n_{\text{color}}$, the circuit performs computational basis
            measurements on all the clock qubits.
            The circuit outputs the AND of all pairwise NAND of the measurement
            outcomes.
    \end{enumerate}
  \end{enumerate}

  Next, we argue that the circuit $V$ can be implemented as a $\QW[1]$
  verification circuit where the big gate is an AND gate at the end of the
  circuit.
  First, we note that the sampling of the integer $h$ can be done using a
  constant size quantum circuit and computational basis measurement.
  We can fanout the measurement outcomes to control the later parts in the
  circuit.
  Second, as the Hamiltonian terms in each group $G^{(h)}$ act on different
  qubits for all $h=0, 1, \ldots, n_{\text{color}}-1$, the measurements
  $\{M_{j,0}, M_{j,1}\}$ can be implemented in parallel.
  These measurements output $x_{h}$, an $m_{h}$-bit vector of classical
  information.
  For $h = n_{\text{color}}$, the circuit first measures all the clock qubits
  and computes the pairwise NAND of the outcome.
  We denote this vector of classical bits as $x_{n_{\text{color}}}$, its length
  is $m_{n_{\text{color}}}$.
  So far, all gates involved are constant size quantum circuits and the
  classical fanout gates.
  Finally, the output of the circuit $V$ is the AND of $x_{h}$ for the sampled
  integer $h$.
  It is easy to reuse the AND gate in all $n_{\text{color}} + 1$ cases as we can
  use fanout of input $1$ to pad short $x_{h}$'s so that they all have length
  $m_{\max}$.
  Then we use controlled SWAP gates to move the bits in $x_{h}$ to the same
  register that can hold $m_{\max}$ qubits and output their AND.

  To complete the proof, we will relate the acceptance probability of $V$ to the
  promise conditions we have for the Hamiltonian problem.
  Notice that when $h = n_{\text{color}}$, the circuit accepts with probability
  \begin{equation*}
    \bra{\psi} \Bigl( \sum_{x: \abs{x} \le 1} \ket{x}\bra{x} \Bigr) \ket{\psi}
    = \bra{\psi} \prod_{k,l} (I-\ket{11}\bra{11})_{k,l} \ket{\psi}
    = \bra{\psi} \prod_{j} \bigl( I - H^{(n_{\text{color}})}_{j}
    \bigr) \ket{\psi}.
  \end{equation*}
  Therefore, we can write the overall probability that this circuit accepts as
  \begin{equation}
    \label{eq:qwone-1}
    \Pr(V \text{ accepts}) =
    \frac{1}{n_{\text{color}} + 1} \sum_{h=0}^{n_{\text{color}}}
    \bra{\psi} \bigotimes_{j=1}^{m_{h}} (I - H^{(h)}_{j})  \ket{\psi}.
  \end{equation}

  In the yes case, the Hamiltonian has ground state energy at most $a$, which
  means that there is a witness state $\ket{\psi}$
  \begin{equation*}
    \bra{\psi} H \ket{\psi} = \bra{\psi} \sum_{j=1}^{m} H_{j} \ket{\psi} \le a.
  \end{equation*}
  Hence, continuing on \cref{eq:qwone-1}, we have
  \begin{equation*}
    \begin{split}
      \Pr(V \text{ accepts})
      & \ge \frac{1}{n_{\text{color}}+1}
        \sum_{h=0}^{n_{\text{color}}} \bra{\psi}
        \Bigl(I - \sum_{j=1}^{m_{h}} H^{(h)}_{j}\Bigr) \ket{\psi}\\
      & = 1 - \frac{\bra{\psi} H \ket{\psi}}{n_{\text{color}}+1}
        \ge 1 - \frac{a}{n_{\text{color}} + 1},
    \end{split}
  \end{equation*}
  where the inequality follows from \cref{lem:sum-vs-prod}.

  In the no case, we have for all state $\ket{\psi}$ of certain weight
  \begin{equation*}
    \bra{\psi} H \ket{\psi} = \bra{\psi} \sum_{j=1}^{m} H_{j} \ket{\psi} \ge b.
  \end{equation*}
  So from \cref{eq:qwone-1}, this gives
  \begin{equation*}
    \begin{split}
      \Pr(V \text{ accepts})
      & \le \frac{1}{n_{\text{color}}+1}
        \sum_{h=0}^{n_{\text{color}}} \bra{\psi}
        \Bigl(I - \frac{1}{m_{h}} \sum_{j=1}^{m_{h}} H^{(h)}_{j}\Bigr)
        \ket{\psi} \\
      & \le \frac{1}{n_{\text{color}}+1}
        \sum_{h=0}^{n_{\text{color}}} \bra{\psi}
        \Bigl(I - \frac{1}{m_{\max}} \sum_{j=1}^{m_{h}} H^{(h)}_{j}\Bigr)
        \ket{\psi} \\
      & = 1 - \frac{\bra{\psi} H \ket{\psi}}{m_{\max}(n_{\text{color}}+1)}
        \le 1 - \frac{b}{n^{2}(n_{\text{color}}+1)},
    \end{split}
  \end{equation*}
  where the first inequality follows from \cref{lem:sum-vs-prod}.
  That is, we can choose
  \begin{equation*}
  c = 1- \frac{a}{n_{\text{color}}+1},\quad s = 1 - \frac{b}{n^{2}(n_{\text{color}}+1)}.
  \end{equation*}
  The condition on the gap
  $c-s = (b/n^{2} - a)/(n_{\text{color}} + 1) \ge 1/\poly(n)$ follows from the
  strong gap condition on $a,b$ for the Hamiltonian problem.
\end{proof}

From this proof we can also conclude that \problem{Weight-$k$ $\ell$-Local
  Hamiltonian} and {\problem{Weight-$k$ Weight-Preserving Quantum Circuit
    Satisfiability}} can be reduced to each other.

\begin{corollary}\label{cor:WLH-WPQCS}
  Given $a,b$ with $b-a>1/\mathrm{poly}(n)$,\problem{Weight-$k$ $\ell$-Local
    Hamiltonian$(a,b)$} reduces to \problem{Weight-$k$ Weight-Preserving Quantum
    Circuit Satisfiability$(c,s)$} under \FPT reduction for some $c,s$ such that
  $c-s>1/\poly(n)$.
  The same is true when reducing \problem{Weight-$k$ Weight-Preserving Quantum
    Circuit Satisfiability$(c,s)$} to \problem{Weight-$k$ $\ell$-Local
    Hamiltonian$(a,b)$}.
\end{corollary}

\begin{proof}
  That \problem{Weight-$k$ $\ell$-Local Hamiltonian$(a,b)$} reduces to
  \problem{Weight-$k$ Weight-Preserving Quantum Circuit Satisfiability$(c,s)$}
  has been already shown.
  It has been shown also that \problem{Weight-$k$ Weight-Preserving Quantum
    Circuit Satisfiability$(c,s)$} reduces to almost spatially sparse weighted
  Local Hamiltonians.
\end{proof}

Finally combining the beyond sections together, we could provide a proof for
\cref{thm:weighted-ham}.
\begin{proof}
  By \cref{lem:wLH-to-wpQCSAT}, given a \problem{weight-$k$ local Hamiltonian
    $(a,b)$} instance $H=\sum_{j=1}^m H_j$ on $n$ qubits with $b-a>1/\poly(n)$,
  we can obtain a \problem{weight-$k$ weight-preserving quantum circuit
    satisfiability} instance $W$ with size $O(km\poly(n))=O(k\poly(n))$, acting
  on $O(n+M+k)=\poly(n)+k$ qubits, completeness $1-\frac{m+a}{M}$ and soundness
  $1-\frac{m+b}{M}$.

  Now we can apply \cref{col:FastwpMW} to amplify the gap to
  $\left(2^{-n},1-2^{-n}\right)$, and the new circuit has size
  $|\mathcal{C}|=O\left(\frac{m}{b-a}|W|\log(2^n)\right)=O(k\poly(n))$ acting on
  $n'=\poly(n)+k$ qubits.
  Using the parameters in \cref{col:wpOT}, we can construct a weight-$2k'+1$
  almost spatially sparse local Hamiltonian instance $H_{\text{sp}}$ with following
  parameters: $k'=k+O(1)$, $T\leq 3n'(|\mathcal{C}|+1)=O(k^2\poly(n))$,
  $a=\frac{1}{(T+1)2^n}$, $b=\frac{c(1-2^{-n/2}-2^{-n})}{T^3}$.
  The Hamiltonian $H_{\text{sp}}$ acts on $n_f=O(k^2\poly(n))$ qubits.

  Finally we apply \cref{lem:almost-sparse-to-circ} to obtain our final $\QW[1]$
  circuit.
  We can check that the energy thresholds $a,b$ we obtained in the step beyond
  satisfies $b/n_f^2-a\geq1/\poly(n)$.
  Thus our $\QW[1]$ circuit constructed in \cref{lem:almost-sparse-to-circ} has
  probability gap $c-s\geq1/\poly(n)$ since $k\leq n$.
\end{proof}

\section{Frustration-Free Weighted Hamiltonian Problems}\label{sec:ff}

In Hamiltonian complexity theory, there is a variant of the local Hamiltonian
problems with physical relevence called frustration-free Hamiltonian problems.
A Hamiltonian $H = \sum_{j} H_{j}$ is frustration-free if its ground state
$\ket{\psi}$ has the lowerest possible energy for each term $H_{j}$.
That is, $\bra{\psi} H_{j} \ket{\psi} = \lambda_{\min}(H_{j})$.
In this case, it is convenient to shift the spectrum of the local terms so that
$H_{j} \ge 0$ and require that $\bra{\psi} H_{j} \ket{\psi} = 0$.
We define a weighed version of the frustration-free Hamiltonian problem as
follows.

\begin{definition}[\problem{Frustration-Free Weight-$k$ $\ell$-Local
  Hamiltonian Problem}]\label{def:ff-WLH}\hfill
  \begin{description}
    \item[Instance:] A local Hamiltonian $H = \sum_{j} H_{j}$ on $n$ qubits
          and a real number $b \ge 1/\poly(n)$.
          Each term $H_{j}$ acts non-trivially on at most $\ell$ qubits
          and satisfies that $0 \le H_{j} \le I$ for all $j$.
    \item[Parameter:] A natural number $k$.
    \item[Yes:] There exists an $n$-qubit weight-$k$ quantum state
          $\ket{\psi}$, such that $\matrixel{\psi}{H_{j}}{\psi} = 0$ for all $j$.
    \item[No:] For all $n$-qubit weight-$k$ quantum state $\ket{\psi}$,
          $\matrixel{\psi}{H}{\psi} \geq b$.
  \end{description}
\end{definition}

It is evident that the frustration-free weighted Hamiltonian problems are
equivalent to the weighted quantum satisfiability problems defined below.

\begin{definition}[\problem{Weight-$k$ Quantum $\ell$-SAT Problem}]\label{def:WQSAT}\hfill
  \begin{description}
    \item[Instance:] A set of projectors $\Pi_{j}$ for $j=1, 2, \ldots, m$ and
          a real number $b \ge 1 / \poly(n)$.
          Each term $\Pi_{j}$ acts on at most $\ell$ qubits.
    \item[Parameter:] A natural number $k$.
    \item[Yes:] There exists an $n$-qubit weight-$k$ quantum state
          $\ket{\psi}$, such that $\Pi_{j}\ket{\psi} = 0$ for all $j$.
    \item[No:] For all $n$-qubit weight-$k$ quantum state $\ket{\psi}$,
          $\sum_{j=1}^{m} \matrixel{\psi}{\Pi_{j}}{\psi} \geq b$.
  \end{description}
\end{definition}

We will show that the weighed quantum SAT problems are complete for
$\SQW_{1}[1]$, a variant of $\QW_{1}[1]$.

\begin{definition}[\problem{Special Weight-$k$ Weft-$1$ Depth-$d$ Quantum Circuit
  Satisfiability}]\hfill
  \begin{description}
    \item[Instance:] A weft-$1$ depth-$d$ quantum circuit $\mathcal{C}$ on $n$
          witness qubits and $\poly(n)$ ancilla qubits where the only big gate is an
          AND gate and it is the last gate of the circuit.
    \item[Parameter:] A natural number $k$.
    \item[Yes:] There exists an $n$-qubit weight-$k$
          quantum state $\ket{\psi}$, such that
          \begin{equation*}
            \Pr[\text{$\mathcal{C}(\ket{\psi})$ accepts}] \geq c.
          \end{equation*}
    \item[No:]
          For every $n$-qubit weight-$k$ quantum state $\ket{\psi}$,
          $\Pr[\text{$\mathcal{C}(\ket{\psi})$ accepts}] \leq s$.
  \end{description}
\end{definition}

\begin{definition}\label{def:SQW11}
  The class $\SQW[1]$ consists of all parameterized problems that are FPQT
  reducible to \problem{Special Weight-$k$ Weft-$1$ Depth-$d$ Quantum Circuit
    Satisfiability} for some constant $d$ and completeness and soundness $c,s$
  satisfying $c-s\ge 1/\poly(n)$.
  The class $\SQW_{1}[1]$ consists of all parameterized problems that are FPQT
  reducible to \problem{Special Weight-$k$ Weft-$1$ Depth-$d$ Quantum Circuit
    Satisfiability} for some constant $d$, $c=1$, and $s \le 1 - 1/\poly(n)$.
\end{definition}

It is obvious that $\W[1] \subseteq \SQW[1] \subseteq \QW[1]$ and
$\W[1] \subseteq \SQW_{1}[1] \subseteq \QW_{1}[1]$.
This is because the big AND gate can be simulated by a big Toffoli gate.

\begin{theorem}\label{thm:SQW1-Frustration-QSAT}
  \problem{Weight-$k$ Quantum $\ell$-SAT} problem and \problem{Frustartion-Free
    Weight-$k$ $\ell$-Local Hamiltonian} problem are complete problems for
  $\SQW_{1}[1]$ for some constant $\ell$.
\end{theorem}

\begin{proof}
  The fact that these two problems are in $\SQW_{1}[1]$ follows from the proof
  that the weighted local Hamiltonian problem is in $\SQW[1]$ and that perfect
  completeness is preserved in the chain of reductions from the local
  Hamiltonian to the $\SQW[1]$ circuit problem.

  We now prove that the weighed quantum $\ell$-SAT problem is $\SQW_{1}[1]$-hard.
  Let $V$ be the constant-depth circuit representing an $\SQW_{1}[1]$ circuit.
  The decision of the circuit is made by first measuring some or all output
  qubits and then outputing the AND of the measurement outcomes.

  The weighted quantum SAT problem we construct consists of $n$ projectors
  \begin{equation*}
    \Pi_{j} = V^{\dagger} (\ket{0}\bra{0}_{j} \otimes I) V.
  \end{equation*}

  Let $\ket{\psi}$ be the witness state to the verification circuit and assume
  without loss of generality that the first $n$ qubits are measured.
  The acceptance probability of the circuit is then
  \begin{equation*}
    \Pr(V \text{ accepts}) = \bigbra{\psi, 0^{\ell}} V^{\dagger}
    (\ket{1^{n}}\bra{1^{n}} \otimes I) V \bigket{\psi, 0^{\ell}}.
  \end{equation*}
  It is straightforward to rewrite it as
  \begin{equation*}
    \Pr(V \text{ accepts}) = \bigbra{\psi, 0^{\ell}} \,
    \prod_{j=1}^{n} (I- \Pi_{j})\, \bigket{\psi, 0^{\ell}}.
 \end{equation*}

  Then, for yes-cases, we have
  \begin{equation*}
    \begin{split}
      \bigbra{\psi, 0^{\ell}} \, \frac{1}{n} \sum_{j=1}^{n} \Pi_{j}\,
      \bigket{\psi, 0^{\ell}}
      & = 1 - \bigbra{\psi, 0^{\ell}}\, (I - \frac{1}{n} \sum_{j=1}^{n} \Pi_{j})\,
        \bigket{\psi, 0^{\ell}}\\
      & \le 1 - \bigbra{\psi, 0^{\ell}} \, \prod_{j=1}^{n} (I- \Pi_{j})\,
        \bigket{\psi, 0^{\ell}}\\
      & = 1 - \Pr(V \text{ accepts}) = 0.
    \end{split}
  \end{equation*}
  In the above equation, the inequality follows from \cref{lem:sum-vs-prod}.
  As each $\Pi_{j}$ is a projector, this implies that
  $\bigbra{\psi,0^{\ell}}\, \Pi_{j}\, \bigket{\psi,0^{\ell}} = 0$.

  For the no-cases,
  \begin{equation*}
    \begin{split}
      \bigbra{\psi, 0^{\ell}} \, \sum_{j=1}^{n} \Pi_{j}\,
      \bigket{\psi, 0^{\ell}}
      & = 1 - \bigbra{\psi, 0^{\ell}}\, (I - \sum_{j=1}^{n} \Pi_{j})\,
        \bigket{\psi, 0^{\ell}}\\
      & \ge 1 - \bigbra{\psi, 0^{\ell}} \, \prod_{j=1}^{n} (I- \Pi_{j})\,
        \bigket{\psi, 0^{\ell}}\\
      & = 1 - \Pr(V \text{ accepts}) \ge 1/\poly(n).
    \end{split}
  \end{equation*}

  To complete the reduction, we need to show that the projectors $\Pi_{j}$ act
  on constant number of qubits.
  We prove this using a light-cone argument.

  \begin{figure}
    \centering

    \begin{tikzpicture}[node distance=1cm,
      gate/.style={draw, minimum size=16, fill=white, inner sep=2,
        minimum height=.8cm, minimum width=.8cm},
      cgate/.style={draw, minimum size=16, fill=ChannelColor, inner sep=2},
      ctrl/.style={circle, fill=black, minimum size=4, inner sep=0},
      octrl/.style={circle, draw=black, fill=white, minimum size=4, inner sep=0},
      target/.style={circle, draw, minimum size=8, inner sep=0},
      brace/.style={decorate, decoration = {calligraphic brace, amplitude=4}},
      measure/.style={draw, minimum height=14, minimum width=16, fill=ChannelColor},
      readout/.style={draw, minimum height=9, minimum width=12, fill=white},
      crossx/.style={path picture={\draw[inner sep=0pt]
          (path picture bounding box.south east) --
          (path picture bounding box.north west)
          (path picture bounding box.south west) --
          (path picture bounding box.north east);}},
      box/.style={align=center, draw, fill=white, minimum height=1.4cm, minimum width=.8cm},
      big/.style={draw, dashed, minimum height=3.6cm, minimum width=1.4cm},
      cbox/.style={align=center, draw, fill=ChannelColor, minimum height=1.4cm, minimum width=.8cm}]

      \node (In0) at (0,0) {};
      \node[above of=In0] (In1) {};
      \node[above of=In1] (In2) {};
      \node[above of=In2] (In3) {};
      \node[above of=In3] (In4) {};
      \node[above of=In4] (In6) {};
      \node[above of=In6] (In7) {};
      \node[above of=In7] (In8) {};
      \node[above of=In8] (In5) {};

      \node (Out0) at (8.5cm,0) {};
      \node[above of=Out0] (Out1) {};
      \node[above of=Out1] (Out2) {};
      \node[above of=Out2] (Out3) {};
      \node[above of=Out3] (Out4) {};
      \node[above of=Out4] (Out6) {};
      \node[above of=Out6] (Out7) {};
      \node[above of=Out7] (Out8) {};
      \node[above of=Out8] (Out5) {};

      \draw (In0) -- (Out0)
      (In1) -- (Out1)
      (In2) -- (Out2)
      (In3) -- (Out3)
      (In4) -- (Out4)
      (In5) -- (Out5)
      (In6) -- (Out6)
      (In7) -- (Out7)
      (In8) -- (Out8);

      \node[cbox] at ([xshift=-1cm]$(Out2) !.5! (Out3)$) {$V_{6}$};
      \node[box] at ([xshift=-1cm]$(Out0) !.5! (Out1)$) {};

      \node[ctrl] (C) at ([xshift=-1cm]Out5) {};
      \node[gate] (T) at ([xshift=-1cm]Out4) {};
      \draw (C) -- (T.north);

      \node[cbox] at ([xshift=-2.2cm]$(Out1) !.5! (Out2)$) {$V_{5}$};
      \node[ctrl] (C) at ([xshift=-3.4cm]Out5) {};
      \node[cbox] (T) at ([xshift=-3.4cm]$(Out2) !.5! (Out3)$) {$V_{4}$};
      \draw (C) -- (T.north);

      \node[gate] (T) at ([xshift=-2.2cm]Out4) {};
      \node[octrl] (C) at ([xshift=-2.2cm]Out7) {};
      \draw (C) -- (T.north);

      \node[gate] (T) at ([xshift=-4.2cm]Out7) {};
      \node[octrl] (C) at ([xshift=-4.2cm]Out8) {};
      \draw (C) -- (T.north);

      \node[ctrl] (C) at ([xshift=-5.4cm]Out6) {};
      \node[target] (T) at ([xshift=-5.4cm]Out5) {};
      \draw (C) -- (T.north);

      \node[ctrl] (C) at ([xshift=-6cm]Out6) {};
      \node[target] (T) at ([xshift=-6cm]Out8) {};
      \draw (C) -- (T.north);

      \node[cbox] at ([xshift=-7.5cm]$(Out3) !.5! (Out2)$) {$V_{2}$};
      \node[cbox] at ([xshift=-5.7cm]$(Out1) !.5! (Out2)$) {$V_{3}$};
      \node[cbox] at ([xshift=-7.5cm]$(Out6) !.5! (Out4)$) {$V_{1}$};

      \node[big] at ([xshift=-5.7cm]$(Out7) !.5! (Out8)$) {};

      \node[label=left:$\ket{0}$] at (In8) {};
      \node[label=left:$\ket{0}$] at (In5) {};

      \node[measure] (M) at ([xshift=.5cm]Out3) {};
      \node[readout] (readout) at (M) {};
      \draw[thick] ($(M) + (1.5mm,-.5mm)$) arc (0:180:1.5mm);
      \draw[thick] ($(M) + (1.5mm,1mm)$) -- ($(M) + (0,-.7mm)$);
      \draw[fill=black] ($(M) + (0,-.7mm)$) circle (0.3mm);

      \draw (Out3.west) -- (M.west);

      \node[left = 1.5cm of In0] (L1) {qubit $1$};
      \node[above of=L1] (L2) {qubit $2$};
      \node[above of=L2] (L3) {qubit $3$};
      \node[above of=L3] (L4) {qubit $4$};
      \node[above of=L4] (L5) {qubit $5$};
      \node[above of=L5] (L6) {qubit $6$};
      \node[above of=L6] (L7) {qubit $7$};
      \node[above of=L7] (L8) {qubit $8$};
      \node[above of=L8] (L9) {qubit $9$};
    \end{tikzpicture}

    \caption{An illustration of computing the light cone of a quantum circuit
      with classical fanout gates.
      In the exmaple, there is a classical fanout gate in the dashed box.
      The qubits $2, 3, 4, 5, 6$ and gates $V_{1}$ to $V_{6}$ are in the light
      cone.}\label{fig:light-cone}
  \end{figure}
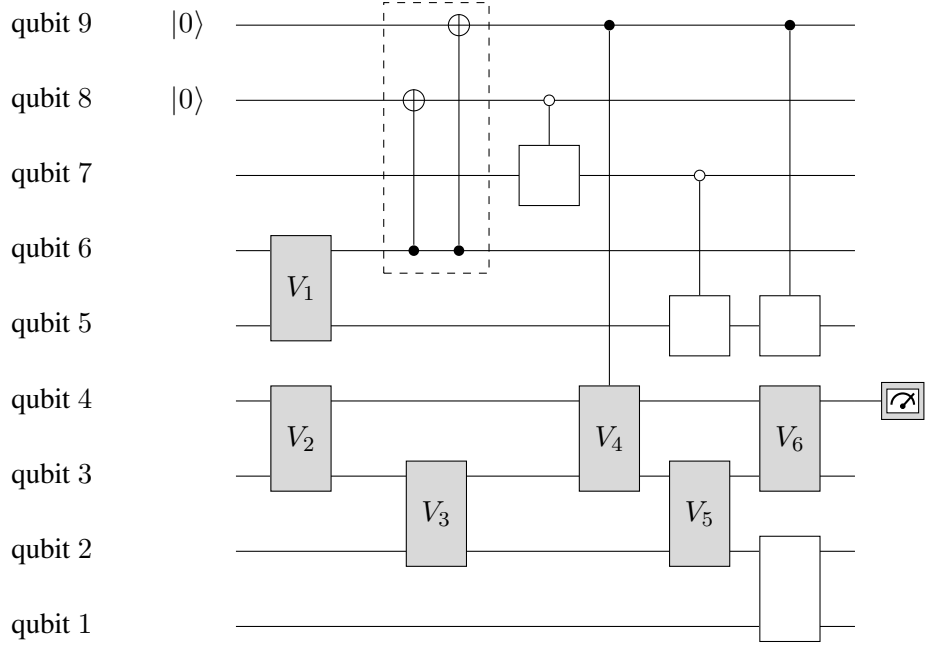

  For a quantum circuit $Q$ with classical fanouts, we define the light cone of
  an output qubit to be the gates and qubits that output qubit $j$ depends on.
  More precisely, we need to model quantum circuits as a network of gates and
  wires and the light cone of a qubit is defined inductively by tracing the
  circuit backwards from the output wire as follows.
  Special care needs to be taken when the gate is a classical fanout gate.
  As $Q$ is a constant depth circuit, it is posssible to partition the gates
  into $d$ layers of gates, each consisting of non-overlapping unitary gates.
  To simplify the presentation, we notice that if a gate uses classical fanout
  input as in the gate $V_{4}$ in \cref{fig:light-cone}, the measurement on the
  witness state that the circuit represents will not change if we change the
  fanout-target control qubit (qubit $9$) of $V_{4}$ to qubit $6$, the control
  qubit of the fanout gate.
  This change may increase the depth but will not change the final measurement
  outcome of the circuit by the principle of delayed measurements.
  We implement this type of change for such cases as a preprocessing step and
  let the resulting circuit be $Q'$.
  By the above discussion, $Q$ and $Q'$ define the same measurement $\Pi_{j}$.
  We now consider the light cone of qubit $j$ for $Q'$ inductively.
  In the base case, only the output qubit wire $j$ is in the light cone.
  If we find a gate that is not a classical fanout gate and has an output wire
  in the light cone, we include all of its input wires to the light cone.
  As when we trace back each of the $d$ levels of $Q'$, the number of wires in
  the light cone is at most multiplied by a constant, there will be at most
  $2^{O(d)}$ input qubit wires in the light cone.

  In the above, we see the crucial difference between a classical fanout and a
  quantum fanout.
  For a classical fanout gate, we only incluce control qubit in the light cone
  and do not count the target qubits as the classical information on these
  target qubits can be inferred from the measurement of the control qubit.
  While if a quantum fanout gate is in the light cone, then all of its input
  qubits should be included because of the entangling power of the quantum
  fanout gate.

\end{proof}

Several remarks are in order.
First, we remark that the above proof does not seem to work if we want to show
weighted local Hamiltonian problems are complete for $\SQW[1]$.
The difficulty is due to the fact that we will need a strong gap condition that
the energy thresholds $a = n(1-c)$ and $b = (1-s)$ satisfy
$b - a \ge 1/\poly(n)$.
The condition holds automatically if $c=1$ but may be false if $c<1$ and there
is no easy way to amplify the gap as we are working with constant depth
circuits.

Second, because of the completeness proof of weighted quantum SAT for
$\SQW_{1}[1]$, the problem of whether weighted quantum SAT problems are complete
for $\QW_{1}[1]$ is essentially equivalent to whether there is a normalization
theorem for weft-$1$ quantum circuits like in the classical case.
Classically, the normalization theorem says that any classical weft-$1$ circuit
can be reduced to a circuit where the only big gate is the last gate and it is
the AND gate.
The quantum case is technically challenging.
The simplest case we don't know how to normalize is when the circuit has a big
NAND gate in the end.
In the classical case, the technique is that we can ask the prover for some
extra information like where a $0$ input to the NAND gate is (using a weight-$1$
indicator) that ensures the acceptance by the final NAND gate.
In the quantum case however, the meaurement before the classical big NAND gate
has intrinsic randomness and the prover is not able to predict the place of the
$0$'s beforehand.

Third, in the above proof, the locality of the resuling Hamiltonian depends on
the depth $d$ of the circuit, which is at most $O(2^{d})$.
It is possible to bring down the locality by going throught the chain of
reduction in \cref{sec:LH-qw1} so that it is independent of the depth $d$.

\section{\QW-hierarchy and \ETH}\label{sec:qw-eth}

As mentioned in the introduction, one of the most important uses of
parameterized complexity theory is in the fine-grained complexity analysis.
In particular, there are important connections between $\W[1]$ and the
exponential time hypothesis (\ETH), some of which are presented in the book
Fundamentals of Parameterized Complexity by Downey and Fellows
(2013)~\cite{downey_fundamentals_2013}.
We use the version of ETH that can be found in Section 29.4
of~\cite{downey_fundamentals_2013}.
In what follows we say that a circuit $C$ has total description size $D$ if the
number of inputs and total number of gates are bounded by $D$.

\begin{definition}[Exponential Time Hypothesis (\ETH)]\label{def:eth}
  We define the Exponential Time Hypothesis as follows.
  There is no algorithm with running time $2^{o(n)}$ that decides for a weft $1$
  Boolean circuit $C$ of total description size $n$, whether there is an input
  vector $x$ such that $C(x)=1$.
\end{definition}

Note that this is a weaker definition than the typical one for \ETH{}.
The reason for the slight weakening of the hypothesis is done in order to make
connections between fine-grained complexity and parameterized complexity (see
Chapter 16, in~\cite{flum2006parameterized}).
In this section we shall consider a quantum version of \ETH{} together with a
quantum-classical version.

\begin{definition}[Quantum Exponential Time Hypothesis (\QETH)]\label{def:QETH}
  We define the $\QETH$ as follows.
  For some $c,s$ with $c-s>1/\poly(n)$, there is no \emph{quantum} algorithm
  running in time $2^{o(n)}$ that decides for a weft-$1$ quantum circuit $Q$ of
  total description size $n$ whether (i) there is an input witness state
  $\ket{\psi}$ such that $\Pr\left(Q(\ket{\psi})\text{ accepts}\right)\geq c$ or
  (ii) for all input witness states $\ket{\psi}$, $\Pr \left(Q(\ket{\psi})\text{
      accepts}\right)\leq s$.
\end{definition}

\begin{definition}[Quantum-Classical Exponential Time Hypothesis (QCETH)]\label{def:QCETH}
  We define the \QCETH{} as follows.
  There is no \emph{quantum} algorithm running in time $2^{o(n)}$ that decides
  for a weft-$1$ Boolean circuit $C$ of total description size $n$, whether
  there is an input vector $x$ such that $C(x)=1$.
\end{definition}

We have defined \QETH{} as a hypothesis about \emph{some} pair $c,s$ with
polynomial gap rather than all such pairs $c,s$.
The reason for this choice is that we want to show that if certain problems are
tractable given any polynomial gap, then \QETH{} is false.
This will be evident later in this section.
Nonetheless, we remark that by changing the definition of \QETH,
\cref{prop:QCETH-imp-QETH} would not be affected and \cref{thm:QM-QETH} would
require some minimal modification.
A natural question is the relationship between these two hypothesis just
defined.
We prove first \QETH{} is a weaker statement than \QCETH{}.

\begin{proposition}\label{prop:QCETH-imp-QETH}
\QCETH{} implies \QETH{}.
\end{proposition}

\begin{proof}
  Assume that \QETH is false, then there is a quantum algorithm $\mathcal{A}$
  deciding the problem in \cref{def:QETH}.
  We shall construct a quantum circuit $\mathcal{Q}$ and show that the
  satisfiability problem on $\mathcal{C}$ reduces to the satisfiability problem
  on $\mathcal{Q}$.
  Let $\mathcal{C}$ be a weft-$1$ classical circuit of total description size
  $n$, we can assume that $C$ has $n$ gates of bounded fan-in $f$ with gate
  basis $\{\textnormal{AND, OR, NOT}\}$.
  First, we modify $\mathcal{C}$ into a reversible circuit by adding an ancilla
  bit initialized at $0$ for each AND and OR gate, including the weft-$1$ gates.
  Note that this increases the number of input bits by $n$ since there are at
  most $n$ gates.
  For the fan-out gates in the classical circuit, these can be replaced by
  reversible CNOTs.
  Note that there are at most $f\cdot n$ possible inputs to the bounded fan-in
  gates, which implies we require at most $O(n)$ CNOT gates.
  After this procedure, we end up with at reversible circuit which can be
  transformed easily into a quantum circuit $\mathcal{Q}$ with $O(n)$ inputs and
  $O(n)$ gates and generalized Toffoli for weft-$1$ gates.
  We also include in $\mathcal{Q}$ a procedure to check that the ancilla qubits
  are all set to $\ket{0}$, which requires $O(n)$ measurements and gates.

  Now we show the decision problem in \cref{def:QCETH} with circuit
  $\mathcal{C}$ reduces to the promise problem with circuit $Q$ in
  \cref{def:QETH} with completeness $b=1$ and soundness $a=0$.
  If there is $x\in\{0,1\}^n$ such that $\mathcal{C}(x)=1$, then consider the
  state $\ket{x0^{cn}}$ where $c>0$ and $(c+1)n$ is the number of inputs to
  $\mathcal{Q}$.
  We have then
  $\mathcal{Q}\ket{x0^{cn}}=\sum_{y\in\{0,1\}^{(c+1)n-1}}\beta_y\ket{1y}$, where
  $\beta_y\in\mathbb{C}$ and $\sum_y \abs{\beta_y}^2=1$.
  Letting $\Pi_1^{(0)}=\ketbra{1}{1}$ be the projector onto the state $\ket{1}$
  for the first qubit, we have that
  \begin{align}
    \Pr\left(\mathcal{Q} \text{ accepts } \ket{x0^{cn}}\right)
    &=\norm{\Pi_1^{(0)}\mathcal{Q}\ket{x0^{cn}}}^2 \nonumber\\
    &=\norm{\sum_{y\in\{0,1\}^{(c+1)n-1}}\beta_y\ket{1y}}^2\nonumber\\
    &= 1
  \end{align}
  Suppose now that for all $x\in\{0,1\}^n$, $\mathcal{C}(x)=0$.
  We have that
  $\mathcal{Q}\ket{x0^{cn}}=\sum_{y\in\{0,1\}^{(c+1)n-1}}\beta_{y,x}\ket{0y}$
  and thus $\Pi_1^{(0)}\mathcal{Q}\ket{x0^{cn}}=0$.
  Any state passing the initial verification of the ancillae qubits has the form
  $\ket{\psi}=\sum_{x\in \{0,1\}^n}\gamma_x \ket{x0^{cn}}$, with
  $\sum_x \abs{\gamma_x}^2 =1$.
  Then we have that
  \begin{align}
    \Pr\left(\mathcal{Q} \text{ accepts } \ket{\psi}\right)
    &=\norm{\Pi_1^{(0)}\mathcal{Q}\ket{\psi}}^2\nonumber\\
    &=\norm{\Pi_1^{(0)}\sum_{x,y}\beta_{y,x} \gamma_x \ket{0y}}^2\nonumber\\
    &= 0.
  \end{align}
  This shows the reduction and thus algorithm $\mathcal{A}$ can solve in time
  $2^{o(n)}$ the decision problem for circuit $\mathcal{C}$ and \QCETH is false.
\end{proof}

\subsection{Miniaturized problems and \ETH}

In this subsection we shall introduce miniaturized problems which are a key
ingredient in connecting results from parameterized complexity and \ETH.
First, we define the miniature version of the classical circuit satisfiability
problem and then we will show how it connects to \ETH and \QCETH.
\begin{definition}[$\problem{Mini-CircSAT}_t$]\hfill
  \begin{description}
    \item[Instance:] Positive integers $k$ and $n$ in unary, and a weft $t$
          Boolean circuit $C$ of total description size at most $k\log n$.
    \item[Parameter:] A natural number $k$.
    \item[Problem:] Decide if there is an input binary vector $x$ such that $C(x)=1$.
  \end{description}
\end{definition}

For simplicity, we will refer to $\problem{Mini-CircSAT}_1$ as
$\problem{Mini-CircSAT}$.
The following theorem illustrates the connection between the tractability of
miniature problems and \ETH.

\begin{theorem}[Theorem 29.4.1 in~\cite{downey_fundamentals_2013}]\label{thm:classical-mini-eth}
  \problem{Mini-CircSAT} is in \FPT if and only if \ETH is false.
\end{theorem}

The \problem{Mini-CircSAT} can be then reduced to \problem{Weight-$k$ Independent
  Set} which implies the following theorem.
\begin{theorem}[Section 29.4 of~\cite{downey_fundamentals_2013}]\label{thm:W1-ETH}
If\/ $\W[1]=\FPT$ then \ETH is false.
\end{theorem}
\cref{thm:W1-ETH} establishes a sufficient condition for \ETH to be
false.
In classical parameterized complexity the complexity class $\M[1]$ is defined as
the closure under \FPT reductions of Mini-CircSAT, the claim that this class is
tractable for \FPT algorithms is equivalent to \ETH being false.
\begin{definition}
    Define $\M[t]$ as the set of problems \FPT reducible to $\problem{Mini-CircSAT}_t$
\end{definition}
\begin{theorem}[Restatement of \cref{thm:classical-mini-eth}]
$\M[1]=\FPT$ if and only if \ETH is false.
\end{theorem}
As an aside, it is straightforward to see that the weighted local Hamiltonian
problem is $\W[1]$-hard, which makes unlikely any $\FPT$ algorithms for this
problem as implied by the above theorem.
To prove this we can simply reduce the weighted independent set problem to the
weighted local Hamiltonian problem.

\begin{proposition}\label{prop:indset-LH}
  The \problem{Weight-$k$ Independent Set} problem reduces to the
  \problem{Weight-$k$ Local Hamiltonian problem$(a,b)$} under $\FPT$ reductions,
  for any $a,b$ with $b>a\geq 0$.
\end{proposition}
\begin{proof}
  Let $G=(V,E)$ be a graph with vertex set $V=\{1,2,\cdots,n\}$.
  For each $i\in V$ define a binary variable $x_i$ and the formula
  $\varphi (x_1,\cdots,x_n)=\bigwedge_{(i,j)\in E} (\neg x_i \lor \neg x_j)$.
  $G$ has an independent set of size $k$ if and only if $\varphi$ is satisfiable
  by a bitstring $x=x_1,\cdots x_n$ of Hamming weight $k$.
  We can map $\varphi$ to a Hamiltonian $H=\sum_i H_i$ acting over $n$ qubits,
  for this, consider the one qubit projector over qubit $i$,
  $\Pi_1^{(i)}=\ketbra{1}{1}$.
  We map each term $(\neg x_i \lor \neg x_j)$ to $H_i=\Pi_1^{(i)}\Pi_1^{(j)}$.
  This Hamiltonian $H$ is an instance of the \problem{Weight-$k$ Local
    Hamiltonian} and has a ground state of energy $0$ with weight-$k$ if and
  only if graph $G$ has an independent set of size $k$.
  Note this reduction works as long as the condition over the $a,b$ in the
  proposition is as given.
\end{proof}

It's known that \problem{Weight-$k$ Independent Set} is
$\W[1]$-complete~\cite{downey_fundamentals_2013}, thus this implies that the
weighted local Hamiltonian problem is $\W[1]$-hard.
An immediate consequence is that its unlikely that there are \FPT algorithms for
the weighted Local Hamiltonian as this would imply that \ETH is false by
\cref{thm:W1-ETH}.
As we show in \cref{thm:LH-QETH}, if this problem can be solved by \FPQT
algorithms then this implies that \QCETH is false.

We can trivially generalize \cref{thm:classical-mini-eth} to the quantum
case, in particular we will frame the results in terms of the weighted local
Hamiltonian problem.
We can give a trivial generalization of \cref{thm:classical-mini-eth} as
follows

\begin{theorem}\label{thm:minisat-QCETH}
 $\M[1]\subseteq\FPQT$ iff $\QCETH$ is false.
\end{theorem}

\begin{proof}
  The proof follows from a direct generalization from the proof of
  \cref{thm:classical-mini-eth} in~\cite{downey_fundamentals_2013}.
  If $\QCETH$ is false, then we can solve \problem{Mini-CircSAT} with a quantum
  algorithm in time $2^{o(k\log n)}$ which is an $\FPT$ function, implying that
  $\M[1]\subseteq\FPQT$.

  Let $C$ be a Boolean circuit of weft $1$ and size $N$ and assume there is an
  \FPQT algorithm that solves \problem{Mini-CircSAT} in time $f(k)n^c$ where we
  assume $f$ to be a growing function in $k$.
  We now show that there is an algorithm deciding if $C$ is satisfiable in time
  $2^{o(n)}$.
  Take $k=f^{-1}(N)$ and $n=2^{(N/k)}$, thus, $N=k\log n$.
  In general, $f^{-1}(N)$ will be a growing function of $N$ and thus $N/k=o(N)$.
  We can now consider the circuit $C$ as an instance of \problem{Mini-CircSAT}
  with $k$ and $n$ chosen as before, giving a runtime for the algorithm of
  $f(f^{-1}(N))(2^{N/k})^c=2^{\frac{cN}{k}+\log N}=2^{o(N)}$, thus \QCETH is
  false.
\end{proof}

As shown in \cref{prop:indset-LH}, the weighted independent set
problem can be reduced to the weighted local Hamiltonian problem.
Moreover, as remarked before, \problem{Mini-CircSat} reduces to the weighted
independent set.

This shows the following
\begin{theorem}\label{thm:LH-QCETH}
  If \problem{Weight-$k$ $\ell$-Local Hamiltonian} is in \FPQT then $\QCETH$ is
  false.
\end{theorem}

\begin{proof}
  The \problem{Weight-$k$ Independent Set} reduces to the \problem{Weight-$k$
    $\ell$-Local Hamiltonian}, by hypothesis we can solve instances of the Local
  Hamiltonian problem in \FPQT and thus \problem{Weight-$k$ Independent Set} as
  well.
  By \cref{thm:minisat-QCETH} the result follows.
\end{proof}

\subsection{Miniaturized problems and \QETH}
Now we turn to a result pertaining to $\QETH$ as defined in \cref{def:QETH}.
Let us begin by defining a miniature version of the quantum circuit
satisfiability problem.

We define the miniature version of the quantum circuit satisfiability
$\problem{Mini-QCSAT}_t(a,b)$ and the class $\QM[t]$ as follows

\begin{definition}[$\problem{Mini-QCSAT}_t(c,s)$]\hfill
  \begin{description}
    \item[Instance:] Integers $k$ and $n$ in unary, and weft-$t$ quantum circuit
          $\mathcal{C}$ of description size $k\log n$.
    \item[Parameter:] A natural number $k$.
    \item[Yes:] There exists an input quantum state $\ket{\psi}$, such that
          $\Pr[\text{$\mathcal{C}(\ket{\psi})$ accepts}] \geq c$.
    \item[No:] For every input quantum state $\ket{\psi}$,
          $\Pr[\text{$\mathcal{C}(\ket{\psi})$ accepts}] \leq s$.
 \end{description}
\end{definition}

\begin{definition}
  Define $\QM_{c,s}[t]$ as the set of problems \FPQT-reducible to
  $\problem{Mini-QCSAT}_t(c,s)$ and define $\QM[t]$ as
  \begin{equation*}
    \QM[t] := \bigcup_{\mathclap{\substack{c,s\\ c-s>1/poly(n)}}} \QM_{c,s}[t].
  \end{equation*}
\end{definition}
We denote as Mini-QCSAT$(c,s)$ the problem $\problem{Mini-QCSAT}_1(c,s)$.
Just as in the classical case, we give a theorem connecting the complexity of
\problem{Mini-QCSAT} and $\QETH$ from \cref{def:QETH}.

\begin{theorem}\label{thm:QM-QETH}
  $\QM[1]\subseteq\FPQT$ iff $\QETH$ is false.
\end{theorem}

\begin{proof}
  The argument from \cref{thm:minisat-QCETH} can be repeated.
  First assume $\QETH$ is false, then for all $c,s$ with polynomial gap there is
  an algorithm that solves the quantum circuit satisfiability problem with
  completeness $c$ and soundness $s$ with $c-s>1/\mathrm{poly}(n)$.
  Then, given an instance $C$ of Mini-QCSAT$(c,s)$ we can use this algorithm to
  solve it in time $2^{o(k\log n)}$ which is an \FPT function.

  Now assume that for all $c,s$ with polynomial gap, Mini-QCSAT$(c,s)$ is
  solvable in time $f(k)n^{c_{0}}$ time for some constant $c_{0}>0$.
  Let $C$ be a weft-$1$ circuit of size $N$.
  Set $k=f^{-1}(N)$ and $n=2^{(N/k)}$, which implies $N=k\log n$.
  In general it will be true that $N/k=o(N)$.
  Using the \FPQT algorithm on $C$, we have a running time $2^{o(N)}$ which
  solves the decision problem with completeness $c$ and soundness $s$.
  Since this is true for all $c,s$ such that $c-s>1/\poly(n)$ then \QETH is
  false.
\end{proof}

Now we show that the Mini-QCSAT reduces to the weight-preserving quantum circuit
satisfiability problem from \cref{def:weightP-wQCSAT}.

\begin{lemma}\label{lem:mini-to-wpreserving}
  \problem{Weight-$k$ Weight-Preserving Quantum Circuit Satisfiability$(c,s)$}
  is $\QM_{c,s}[1]$-hard.
\end{lemma}

\begin{proof}
  Let $\mathcal{C}$ describe a Mini-QCSAT$(c,s)$ circuit with at most $k\log n$
  inputs and $k \log n$ gates.
  We can decompose these gates into one qubit gates and CNOTs, increasing the
  number of gates to $\mathrm{poly}(k\log n)$.
  Note that a $k\log n$ qubit state $\ket{\chi}$ can be mapped to a weight-$k$
  $n$-qubit state $\ket{\psi}$ by considering the natural encoding of an $n$
  qubit state of weight-$k$ with $k\log n$ qubits.
  If $\mathcal{C}$ has less than $k\log n$ input qubits then we can always add
  ancillas in the $\ket{0}$ state, and measure at the end of the circuit to
  check that they are all in the $\ket{0}$ state, we can thus assume that
  $\mathcal{C}$ has $k \log n$ input qubits.

  We take the $k\log n$ input qubits and divide them into $k$ groups of $\log n$
  qubits and consider the encoding of the $\log n$ qubit state into an $n$ qubit
  state of weight-$1$.
  For bitstring $x\in \{0,1\}^{\log n}$ we denote as $E(x)$ the encoding into a
  bitstring of length $n$ and Hamming weight $1$ which preserves lexicographic
  order.
  For example, if $n=4$ we consider the encoding $\ket{E(00)}=\ket{0001}$,
  $\ket{E(01)}=\ket{0010}$, $\ket{E(10)}=\ket{0100}$ and
  $\ket{E(11)}=\ket{1000}$.
  This mapping will result in a circuit with $kn$ input qubits.
  We now explain how to map the gates in circuit $\mathcal{C}$ to the encoded
  version $\mathcal{C'}$ in such a way that the weight is preserved in circuit
  $\mathcal{C'}$.
  A one-qubit gate $V$ in $\mathcal{C}$ is mapped to $2^{\log n-1}$ $\hat{V}$
  weight-preserving gates acting over two qubits as in
  \cref{def:U-weightpreserving}.
  The qubits over which these gates act can be computed efficiently.
  Suppose gate $V$ acts over qubit $i\in \{1,2,\ldots,\log n\}$, where the index
  $i$ runs over the qubits inside some group of $\log n$ qubits.
  Denote as $\Tilde{V}$ the encoded version in the new circuit of gate $V$.
  The action of $\Tilde{V}$ over basis states is defined as follows.
  Let $x^{(1)}$,$x^{(2)}\in\{0,1\}^{\log n}$ be computational basis states which
  differ only on the $i$th bit, for example, suppose $x^{(1)}=0$ and
  $x^{(2)}_i=1$.
  Let $p$ and $q$ be the qubit indices in the new circuit where $E(x^{(1)})$ and
  $E(x^{(2)})$ have a $1$.
  Then, $\Tilde{V}$ will act as gate $\hat{V}$ on qubits $p$ and $q$.
  For each such pair $x^{(1)}$ and $x^{(2)}$, $\Tilde{V}$ acts on the prescribed
  pair of qubits as $\hat{V}$.
  Thus, in total $\Tilde{V}$ requires $2^{\log n-1}$ $\hat{V}$ gates.
  An example is illustrated in \cref{fig:klogn-trick} where $n=8$ and
  $k=1$, each group has $3$ qubits and is encoded as a group of $8$ qubits.

  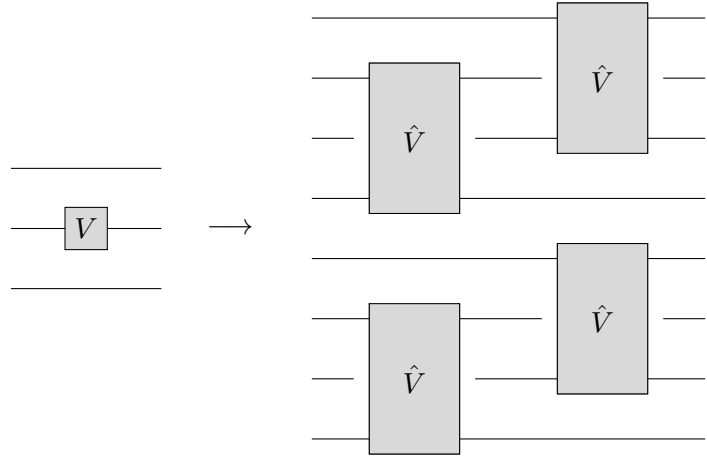
\begin{figure}[ht]
  \centering
  \begin{tikzpicture}[node distance=0.8cm,
    gate/.style={draw, minimum size=16, fill=ChannelColor, inner sep=2},
    ctrl/.style={circle, fill=black, minimum size=4, inner sep=0},
    octrl/.style={circle, draw=black, fill=white, minimum size=4, inner sep=0},
    target/.style={circle, draw, minimum size=8, inner sep=0},
    brace/.style={decorate, decoration = {calligraphic brace, amplitude=4}},
    measure/.style={draw, minimum height=14, minimum width=16, fill=gray},
    readout/.style={draw, minimum height=9, minimum width=12, fill=white},
    crossx/.style={path picture={\draw[inner sep=0pt]
        (path picture bounding box.south east) --
        (path picture bounding box.north west)
        (path picture bounding box.south west) --
        (path picture bounding box.north east);}},
    box/.style={align=center, draw, fill=ChannelColor, minimum height=2cm, minimum width=1.2cm},
    ]

    \begin{scope}
      \node (In0) at (0,0) {};
      \node[above of=In0] (In1) {};
      \node[above of=In1] (In2) {};
      \node[above of=In2] (In3) {};

      \node (Out0) at (5.5cm,0) {};
      \node[above of=Out0] (Out1) {};
      \node[above of=Out1] (Out2) {};
      \node[above of=Out2] (Out3) {};

      \node[box] (B1) at ([xshift=1.5cm]$(In1)$) {$\hat{V}$};
      \node[box] (B3) at ([xshift=4cm]$(In2)$) {$\hat{V}$};

      \draw (In0) -- (Out0)
      (In1) -- ([xshift=-2mm]B1.west)
      (Out1) -- ([xshift=2mm]B1.east)
      (In2) -- ([xshift=-2mm]B3.west)
      (Out2) -- ([xshift=2mm]B3.east)
      (In3) -- (Out3);

      \node[box] (B1) at ([xshift=1.5cm]$(In1)$) {$\hat{V}$};
      \node[box] (B3) at ([xshift=4cm]$(In2)$) {$\hat{V}$};
    \end{scope}

    \begin{scope}[yshift=3.2cm]
      \node (In0) at (0,0) {};
      \node[above of=In0] (In1) {};
      \node[above of=In1] (In2) {};
      \node[above of=In2] (In3) {};

      \node (Out0) at (5.5cm,0) {};
      \node[above of=Out0] (Out1) {};
      \node[above of=Out1] (Out2) {};
      \node[above of=Out2] (Out3) {};

      \node[box] (B1) at ([xshift=1.5cm]$(In1)$) {$\hat{V}$};
      \node[box] (B3) at ([xshift=4cm]$(In2)$) {$\hat{V}$};

      \draw (In0) -- (Out0)
      (In1) -- ([xshift=-2mm]B1.west)
      (Out1) -- ([xshift=2mm]B1.east)
      (In2) -- ([xshift=-2mm]B3.west)
      (Out2) -- ([xshift=2mm]B3.east)
      (In3) -- (Out3);

      \node[box] (B1) at ([xshift=1.5cm]$(In1)$) {$\hat{V}$};
      \node[box] (B3) at ([xshift=4cm]$(In2)$) {$\hat{V}$};
    \end{scope}

    \begin{scope}[yshift=2cm]
      \node (In0) at (-4cm,0) {};
      \node[above of=In0] (In1) {};
      \node[above of=In1] (In2) {};

      \node[right=2cm of In0] (Out0) {};
      \node[above of=Out0] (Out1) {};
      \node[above of=Out1] (Out2) {};

      \draw (In0) -- (Out0)
      (In1) -- (Out1)
      (In2) -- (Out2);
      \node[gate] at ($(In1) !.5! (Out1)$) {$V$};

      \node[label=right:{$\longrightarrow$}] at ([xshift=.2cm]Out1) {};
    \end{scope}
  \end{tikzpicture}
  \caption{Example of mapping a one-qubit gate to gates acting on $8$ qubits for
    $n=8$ and $k=1$.
    The discontinued lines are qubits that are not acted upon by the gates.}
  \label{fig:klogn-trick}
  \end{figure}

  It is simple to check that this new circuit preserves the amplitudes of the
  original miniature circuit.
  Let $x=(x_1,x_2,\ldots,x_i,\ldots,x_{\log n})\in \{0,1\}^{\log n}$ and
  $V^{(a)}=\sum_{r,s=0}^1 v_{r,s}\ketbra{r}{s}$ the single-qubit unitary acting
  on qubit $i$.
  Then, the action of $V^{(i)}$ over a computational basis state is
  \begin{equation*}
    V^{(i)}\ket{x_1 \cdots x_i \cdots x_{\log n}}
    = v_{0,x_i} \ket{x_1\cdots 0_i \cdots x_{\log n}}
    + v_{1,x_i} \ket{x_1\cdots 1_i \cdots x_{\log n}}.
  \end{equation*}
  where $0_i$ and $1_i$ denote a $0$ or a $1$ at the $i$th position
  respectively.
  The encoded version of $V$ will act in a similar way by construction
  \begin{equation*}
    \Tilde{V}^{(i)}\ket{E(x_1 \cdots x_i \cdots x_{\log n})}
    = v_{0,x_i} \ket{E(x_1 \cdots 0 \cdots x_{\log n})}
    + v_{1,x_i}\ket{E(x_1 \cdots 1 \cdots x_{\log n})}.
  \end{equation*}

  For CNOT gates we need to consider two different cases, (i) the CNOT is acting
  between two qubits in the same group and (ii) the CNOT is acting between two
  qubits in different groups.

  For case (i), suppose CNOT acts on control qubit $i$ and target qubit $j$
  where $i$ and $j$ are in the same group of $\log n$ qubits.
  Let $x^{(1)}, x^{(2)}\in\{0,1\}^{\log n}$, if they differ in the $j$th qubit
  and the $i$th qubit is $1$, then in the new circuit apply a SWAP between the
  qubits where $E(x_1)$ and $E(x_2)$ have 1s.
  We add as many SWAPS as pairs $x^{(1)}, x^{(2)}$ fulfilling this condition
  exist.

  For case (ii), we consider control qubit $i$ and target qubit $j$ such that
  both qubits belong to different groups.
  To implement this gate in the weight-preserving circuit we will require two
  ancillae in the state $\ket{01}$.
  For every $x\in\{0,1\}^{\log n}$ such that qubit $i$ is $1$, then we apply a
  Fredkin gate with control qubit given by the position of $1$ in $E(x)$ and
  with the ancilla qubits as targets.
  Such an example is given in \cref{fig:klogn-cnot}, inside the green box
  the Fredkin gates are applied such that if any of the qubits are in state
  $\ket{1}$ then a SWAP network is applied in the other group, after this the
  action of the Fredkin gates is undone.
  The SWAP network consists of SWAP operators acting over qubits as determined
  by the one-bit case mentioned earlier in our proof, these SWAP gates are
  controlled by the ancilla qubit. 
  
  Note that in the original circuit $\mathcal{C}$ the output is given by a
  single qubit.
  In the new weight-preserving circuit we can add two more extra ancillas in the
  state $\ket{01}$ which we assign as the output qubits.
  Then, after acting with the weight-preserving simulation of $\mathcal{C}$, we
  can act with several controlled SWAP operators with the output qubits as
  target and the control qubits corresponding to those that encode states of
  $\log n$ qubits with the output set to $1$.

  With the mapping in place, we have constructed a weight-preserving circuit and
  the last step is to implement measurements to check that each group of $n$
  qubits has only one qubit set to $\ket{1}$.
  In what follows, let $Q_i=\{q^{(i)}_1,\ldots,q^{(i)}_n\}$ be the set of qubits
  belonging to the $i$th group of qubits, where $i=\{1,\ldots,k\}$.
  To check that each $Q_i$ is a weight-$1$ state, we can include $k(n+1)$
  ancillas, which we also group into $k$ sets of $n$ qubits and denote as
  $A_j=\{a^{(j)}_1,\ldots,a^{(j)}_{n+1}\}$ the $j$th group of $n$ qubits.
  First, initialize each $A_j$ in the weight-$1$ state $\ket{1000\cdots 0}$.
  Next, we will use the qubits in $A_i$ to count the weight of the state in the
  $Q_i$ register.
  We construct the following weight-preserving circuit acting between sets $A_i$
  and $Q_i$ for each $i\in\{1,\ldots,k\}$.
  Act with a controlled SWAP on qubits $a^{(i)}_1,a^{(i)}_2$ as targets and
  qubit $q^{(i)_1}$ as control which we define as
  $\operatorname{CSWAP}_{q_1,a_1,a_2}$ (for simplicity, we surpress the index
  $i$ from now on).
  Then, act with the gate
  $ \operatorname{CSWAP}_{q_2,a_1,a_2}\cdot\operatorname{CSWAP}_{q_2,a_2,a_3}$.
  We act in the same way with succesive qubits in the set $Q_i$; for each qubit
  $q_j$ we act with
  $\operatorname{CSWAP}_{q_i,a_1,a_2} \cdot\operatorname{CSWAP}_{q_i,a_2,a_3}
  \cdots\operatorname{CSWAP}_{q_j,a_{j},a_{j+1}}$.
  Once applied this circuit, we need to only measure the qubit $a_2$ which tells
  us whether the weight of the state in $Q_i$ is $1$.
  Finally, to measure whether all $a^{(i)}_2$ are in the state $1$, we add two
  ancillas in state $\ket{01}$ and act with $\operatorname{CSWAP}$ controlled by
  each $a^{(i)}_2$ and target the two new ancillas, such that we get the state
  $\ket{10}$ if all $a^{(i)}_2$ are in $1$ are $\ket{01}$ otherwise.
  This construction requires the ancilla states to have in total weight-$(k+1)$
  which together with the input state and two more ancillas for the CNOT gates
  and other two for the output qubits gives a total of weight-$2(k+2)$ with
  $kn+k(n+1)+6$ qubits.
  We have then a new weight-preserving circuit $\mathcal{C}'$ which is
  satisfiable by a weight-$2(k+1)$ state if and only if the original circuit
  $\mathcal{C}$ is satisfiable.
  Moreover, $\mathcal{C}'$ simulates $\mathcal{C}$ faithfully (at each step the
  amplitudes are preserved).
  Let us now show that the reduction works as intended.
  For completeness, since the simulation is faithfull, then our new
  weight-preserving circuit preserves the completeness.
  For soundness, suppose for all states $\ket{\psi}$ we have that
  $\Pr\left(\mathcal{C} \text{ accepts } \ket{\psi}\right)\leq s$.
  Let $\ket{\phi}$ be a $kn+k(n+1)+4$ qubit state, where the witness has been
  supplied by the prover and the ancillas have been set as described above.
  Note that the only way for the prover to cheat is by breaking the encoding we
  have delineated above, thus we introduce the decomposition
  $\ket{\phi}=\alpha \ket{\xi_1}+\beta\ket{\xi_2}$, where $\ket{\xi_1}$ is a
  state respecting the encoding and $\ket{\xi_2}$ is a state that doesn't
  respect the encoding above.
  Thus, defining $\Pi_{\text{10}}$ as the projector on the output qubits onto
  the state $\ket{10}$, we have that
  \begin{align}
    \Pr\left(\mathcal{C'} \text{ accepts } \ket{\phi}\right)
    &= \norm{\Pi_{10}\mathcal{C'}\ket{\phi}}^2\nonumber\\
    &= \abs{\alpha}^2 \norm{\Pi_{10}\mathcal{C'}\ket{\xi_1}}^2
  \end{align}
  Since $\abs{\alpha}^2<1$ when the prover is cheating, then the accepting
  probability only diminishes when this is the case.
  Thus $\Pr\left(\mathcal{C'} \text{ accepts } \ket{\phi}\right)\leq s$.
  This implies the $\QM_{c,s}[1]$-hardness of \problem{Weight-$k$
    Weight-Preserving Quantum Circuit Satisfiability$(c,s)$}
\end{proof}

\begin{figure}[ht]
  \centering
  \begin{tikzpicture}[node distance=0.35cm,
    gate/.style={draw, minimum size=16, fill=ChannelColor, inner sep=2},
    ctrl/.style={circle, fill=black, minimum size=4, inner sep=0},
    octrl/.style={circle, draw=black, fill=white, minimum size=4, inner sep=0},
    target/.style={circle, draw, minimum size=8, inner sep=0},
    brace/.style={decorate, decoration = {calligraphic brace, amplitude=4}},
    measure/.style={draw, minimum height=14, minimum width=16, fill=gray},
    readout/.style={draw, minimum height=9, minimum width=12, fill=white},
    crossx/.style={path picture={\draw[inner sep=0pt]
        (path picture bounding box.south east) --
        (path picture bounding box.north west)
        (path picture bounding box.south west) --
        (path picture bounding box.north east);}},
    box/.style={align=center, draw, fill=ChannelColor, minimum height=3.2cm, minimum width=1.8cm},
    ]

    \begin{scope}
      \node (In0) at (-.5cm,0) {};
      \node[above of=In0] (In1) {};
      \node[above of=In1] (In2) {};
      \node[above of=In2] (In3) {};
      \node[above of=In3] (In4) {};
      \node[above of=In4] (In5) {};
      \node[above of=In5] (In6) {};
      \node[above of=In6] (In7) {};

      \node (Out0) at (5.5cm,0) {};
      \node[above of=Out0] (Out1) {};
      \node[above of=Out1] (Out2) {};
      \node[above of=Out2] (Out3) {};
      \node[above of=Out3] (Out4) {};
      \node[above of=Out4] (Out5) {};
      \node[above of=Out5] (Out6) {};
      \node[above of=Out6] (Out7) {};

      \node[above=.7cm of In7] (A) {};
      \node[above of=A] (B) {};

      \node[above=.7cm of B] (A0) {};
      \node[above of=A0] (A1) {};
      \node[above of=A1] (A2) {};
      \node[above of=A2] (A3) {};
      \node[above of=A3] (A4) {};
      \node[above of=A4] (A5) {};
      \node[above of=A5] (A6) {};
      \node[above of=A6] (A7) {};

      \node[above=.7cm of Out7] (C) {};
      \node[above of=C] (D) {};

      \node[above=.7cm of D] (C0) {};
      \node[above of=C0] (C1) {};
      \node[above of=C1] (C2) {};
      \node[above of=C2] (C3) {};
      \node[above of=C3] (C4) {};
      \node[above of=C4] (C5) {};
      \node[above of=C5] (C6) {};
      \node[above of=C6] (C7) {};

      \draw (In0) -- (Out0)
      (In1) -- (Out1)
      (In2) -- (Out2)
      (In3) -- (Out3)
      (In4) -- (Out4)
      (In5) -- (Out5)
      (In6) -- (Out6)
      (In7) -- (Out7);

      \draw (A) -- (C) (B) -- (D);

      \draw (A0) -- (C0)
      (A1) -- (C1)
      (A2) -- (C2)
      (A3) -- (C3)
      (A4) -- (C4)
      (A5) -- (C5)
      (A6) -- (C6)
      (A7) -- (C7);

      \node[ctrl] (E) at ([xshift=1cm]$(A7)$) {};
      \node[crossx] (F) at ([xshift=1cm]$(A)$) {};
      \node[crossx] at ([xshift=1cm]$(B)$) {};
      \draw (E) -- (F.center);

      \node[ctrl] (E) at ([xshift=5cm]$(A7)$) {};
      \node[crossx] (F) at ([xshift=5cm]$(A)$) {};
      \node[crossx] at ([xshift=5cm]$(B)$) {};
      \draw (E) -- (F.center);

      \node[ctrl] (E) at ([xshift=1.5cm]$(A6)$) {};
      \node[crossx] (F) at ([xshift=1.5cm]$(A)$) {};
      \node[crossx] at ([xshift=1.5cm]$(B)$) {};
      \draw (E) -- (F.center);

      \node[ctrl] (E) at ([xshift=4.5cm]$(A6)$) {};
      \node[crossx] (F) at ([xshift=4.5cm]$(A)$) {};
      \node[crossx] at ([xshift=4.5cm]$(B)$) {};
      \draw (E) -- (F.center);

      \node[ctrl] (E) at ([xshift=2cm]$(A3)$) {};
      \node[crossx] (F) at ([xshift=2cm]$(A)$) {};
      \node[crossx] at ([xshift=2cm]$(B)$) {};
      \draw (E) -- (F.center);

      \node[ctrl] (E) at ([xshift=4cm]$(A3)$) {};
      \node[crossx] (F) at ([xshift=4cm]$(A)$) {};
      \node[crossx] at ([xshift=4cm]$(B)$) {};
      \draw (E) -- (F.center);

      \node[ctrl] (E) at ([xshift=2.5cm]$(A2)$) {};
      \node[crossx] (F) at ([xshift=2.5cm]$(A)$) {};
      \node[crossx] at ([xshift=2.5cm]$(B)$) {};
      \draw (E) -- (F.center);

      \node[ctrl] (E) at ([xshift=3.5cm]$(A2)$) {};
      \node[crossx] (F) at ([xshift=3.5cm]$(A)$) {};
      \node[crossx] at ([xshift=3.5cm]$(B)$) {};
      \draw (E) -- (F.center);

      \node[ctrl] (E) at ([xshift=3cm]$(B)$) {};
      \node[box, text width=1.8cm] (box) at ([xshift=3cm]$(In3) !.5! (In4)$)
      {SWAP Network};
      \draw (E) -- (box.north);

      \node[label=left:{$\ket{0}$}] at (B) {};
      \node[label=left:{$\ket{1}$}] at (A) {};

      \coordinate (O) at ($(A) !.5! (B)$);
      \node[left=1.5cm of O] {$\longrightarrow$};

      \node[above left=1.4cm and 3cm of O] (G0) {};
      \node[below of=G0] (G1) {};
      \node[below of=G1] (G2) {};
      \node[left=1.5cm of G0] (H0) {};
      \node[below of=H0] (H1) {};
      \node[below of=H1] (H2) {};
      \draw (G0) -- (H0) (G1) -- (H1) (G2) -- (H2);

      \node[below left=.7cm and 3cm of O] (M0) {};
      \node[below of=M0] (M1) {};
      \node[below of=M1] (M2) {};
      \node[left=1.5cm of M0] (N0) {};
      \node[below of=N0] (N1) {};
      \node[below of=N1] (N2) {};
      \draw (M0) -- (N0) (M1) -- (N1) (M2) -- (N2);

      \node[ctrl] (E) at ($(G1) !.5! (H1)$) {};
      \node[target] (F) at ($(M1) !.5! (N1)$) {};
      \draw (E) -- (F.south);
    \end{scope}
    \end{tikzpicture}
    \caption{Example of mapping a CNOT gate acting between two different groups
      for $n=8$ and $k=2$.
      The gates in the green box implement the control and the SWAP network
      implements the bit flip part.}
    \label{fig:klogn-cnot}
\end{figure}
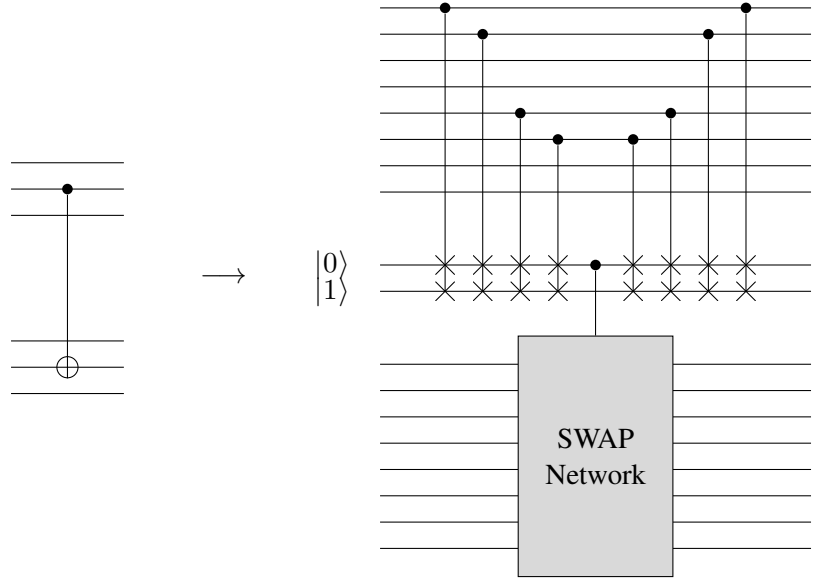

We have thus shown the following corollary.

\begin{corollary}\label{cor:mini-wpqcsat}
  $\problem{Mini-QCSAT}(c,s)$ is in \FPQT if \problem{Weight-$k$
    Weight-Preserving Quantum Circuit Satisfiability}$(c,s)$ is in \FPQT.
\end{corollary}

Now we can use the reduction from the proof of \cref{thm:weighted-ham}
(\cref{cor:WLH-WPQCS}) to reduce the weigh-preserving circuit to an
instance of the almost spatially sparse weight-$k$ local Hamiltonian and thus
also can be reduced to the weight-$k$ $\ell$-local Hamiltonian.

\begin{theorem}\label{thm:LH-QETH}
  If for all $a,b$ such that $b-a>1/\poly(n)$, \problem{Weight-$k$ $\ell$-Local
    Hamiltonian$(a,b)$} is in \FPQT then $\QETH$ is false.
\end{theorem}

\section{Discussion}\label{sec:discussion}

In this paper we have explored the complexity of the weighted local Hamiltonian
problem.
We have proven that this problem is in $\QW[1]$, but it remains a challenging
open question as to whether it is in fact $\QW[1]$-complete.
The obstacle when using techniques based on the clock construction such as
in~\cite{kitaev2002classical,kempe2006complexity} is that when reducing from
\problem{Weight-$k$ Weft-$1$ Depth-$d$ Quantum Circuit Satisfiability} to the
weighted local Hamiltonian problem, the history state is required to be of
weight-$k$.
Recall that the circuit in the original instance is not required to be
weight-preserving and thus applying the clock construction directly does not
work as it takes the history state out of the weight-$k$ subspace.
Another possibility is to apply the reduction used
in~\cite{Childs2014_bosehubbard}, where the authors prove the \QMA-completeness
of the Bose-Hubbard model.
This proof technique for \QMA-hardness is based on using the Bose-Hubbard
Hamiltonian to simulate the quantum circuit of the verifier.
This naturally requires using $O(n)$ particles, yet in our case the number of
excitations would be bounded by $k$.
Although the witness of the weighted quantum circuit satisfiability is also
bounded, it is not guaranteed to remain bounded as, again, the circuit is not
weight preserving.
In the proof of \cref{lem:mini-to-wpreserving}, we have reduced a miniaturized
version of the circuit satisfiability problem to a weight-preserving circuit,
one way to prove completeness might be to extend this technique to the case of
$n$ qubits and constant depth.

An interesting direction in the future would be to study the local Hamiltonian
problem under other parameterizations, one possibility is to consider parameters
over the interaction graph of the Hamiltonian.
We consider our results here as a first step towards a more fine-grained
analysis of the complexity in the local Hamiltonian problem.
In classical complexity theory, parameters such as the treewidth or branchwidth
play a key role in finding $\FPT$ algorithms for graph problems.
Studying the interaction of such parameters (or finding new ones) in the quantum
setting for local Hamiltonian problems may prove to be a fruitful area of
research for finding efficient algorithms.

\section{Acknowledgements}

MJB acknowledges the support of Google.
MJB and MESM were supported by the ARC Centre of Excellence for Quantum
Computation and Communication Technology (CQC2T), project number CE170100012.
This material is also in part based upon work by MJB and LM supported by the
Defense Advanced Research Projects Agency under Contract No.~HR001122C0074.
Any opinions, findings and conclusions or recommendations expressed in this
material are those of the author(s) and do not necessarily reflect the views of
the Defense Advanced Research Projects Agency.
ZJ acknowledges the support of Australian Research Council (DP200100950).
ZJ and XL acknowledge the support of a startup funding from Tsinghua University.
MESM was supported by a scholarship top-up and extension from the Sydney Quantum
Academy.

\bibliographystyle{alpha}

\bibliography{bibliography}

\end{document}